\documentclass{sig-alternate}
\usepackage{graphicx}
\usepackage{amssymb,amsmath}
\usepackage{algorithm}
\usepackage{algorithmic}
\usepackage{subfigure}

\begin{document}

\title{Efficient Truss Maintenance in Evolving Networks}

\numberofauthors{3}

\author{
% You can go ahead and credit any number of authors here,
% e.g. one 'row of three' or two rows (consisting of one row of three
% and a second row of one, two or three).
%
% The command \alignauthor (no curly braces needed) should
% precede each author name, affiliation/snail-mail address and
% e-mail address. Additionally, tag each line of
% affiliation/address with \affaddr, and tag the
% e-mail address with \email.
%
% 1st. author
%
\alignauthor
Rui Zhou\\
%\titlenote{Part of the work was done while the author was visiting Australian National University.}
       \affaddr{Center for Applied Informatics}\\
       \affaddr{College of Engineering \& Science}\\
       \affaddr{Victoria University}\\
%       \affaddr{Melbourne, VIC 3122, Australia}\\
       \affaddr{Melbourne, Australia}\\
       \email{Rui.Zhou@vu.edu.au}
% 2nd. author\alignauthor
\alignauthor
Chengfei Liu\\
       \affaddr{Faculty of Information and Communication Technologies}\\
       \affaddr{Swinburne University of Technology}\\
%       \affaddr{Melbourne, VIC 3122, Australia}\\
       \affaddr{Melbourne, Australia}\\
       \email{cliu@swin.edu.au}% 3rd. author% 3rd. author
\alignauthor
Jeffrey Xu Yu\\
       \affaddr{Department of Systems Engineering \& Engineering Management}\\
       \affaddr{The Chinese University of Hong Kong}\\
       \affaddr{Hong Kong, China}\\
       \email{yu@se.cuhk.edu.hk}
\and
\alignauthor
Weifa Liang\\
       \affaddr{Research School of Computer Science}\\
       \affaddr{Australian National University}\\
%       \affaddr{Canberra, NSW 2601, Australia}\\
       \affaddr{Canberra, Australia}\\
       \email{wliang@cs.anu.edu.au}
 % use '\and' if you need 'another row' of author names
% 4th. author
% There's nothing stopping you putting the seventh, eighth, etc.
\alignauthor
Yanchun Zhang\\
       \affaddr{Centre for Applied Informatics}\\
       \affaddr{College of Engineering \& Science}\\
       \affaddr{Victoria University}\\
       \affaddr{Melbourne, Australia}\\
%       \affaddr{Canberra, Australia}\\
       \email{Yanchun.Zhang@vu.edu.au}
%\alignauthor
%Jianxin Li\\
%       \affaddr{Faculty of Information and Communication Technologies}\\
%       \affaddr{Swinburne University of Technology}\\
%       \affaddr{Melbourne, VIC 3122, Australia}\\
%       \affaddr{Melbourne, Australia}\\
%       \email{jianxinli@swin.edu.au}
}% author on the opening page (as the 'third row') but

\newtheorem{theorem}{Theorem}
\newtheorem{corollary}[theorem]{Corollary}
\newtheorem{example}{Example}
\newtheorem{lemma}{Lemma}
\newtheorem{observation}{Observation}
\newtheorem{conjecture}{Conjecture}
%\newdef{definition}{Definition}
\newtheorem{definition}{Definition}

\newcommand{\includeonefigure}[3]
{
\begin{figure}[t]
    \centering
    #1
    \caption{#2}
    \label{#3}
\end{figure}
}

\maketitle

\begin{abstract}
Truss was proposed to study social network data represented by graphs. A $k$-truss of a graph is a cohesive subgraph, in which each edge is contained in at least $k-2$ triangles within the subgraph. While truss has been demonstrated as superior to model the close relationship in social networks and efficient algorithms for finding trusses have been extensively studied, very little attention has been paid to truss maintenance. However, most social networks are evolving networks. It may be infeasible to recompute trusses from scratch from time to time in order to find the up-to-date $k$-trusses in the evolving networks. In this paper, we discuss how to maintain trusses in a graph with dynamic updates. We first discuss a set of properties on maintaining trusses, then propose algorithms on maintaining trusses on edge deletions and insertions, finally, we discuss truss index maintenance. We test the proposed techniques on real datasets. The experiment results show the promise of our work.

\end{abstract}

\section{Introduction}

Truss was proposed by Jonathan Cohen in~\cite{DBLP:journals/web/truss} in 2008 to identify cohesive subgraphs for social network analysis. The motivation originated from social structure is that if two persons are strongly tied, they should share enough common friends. As such, we have the definition of $k$-truss as follows:
\begin{definition}
A $k$-truss in a graph $G=(V,E)$ is a non-trivial\footnote{A trivial graph is not a 2-truss due to this non-trivial condition.} connected subgraph $G_s=(V_s,E_s)$, where each edge in $E_s$ is contained in at least $k-2$ triangles of $G_s$, here $V_s \subseteq V$ and $E_s \subseteq E$.
\end{definition}
Here, an edge $(a,b)$ is contained in $k-2$ triangles in $G_s$ equals that nodes $a$ and $b$ have $k-2$ common neighbors in $G_s$. Usually, we are interested in the maximal $k$-trusses, i.e. the $k$-trusses which are not proper subgraphs of other $k$-trusses, otherwise there may be too many overlapping trusses resulting from nodes combination. 

Consider a social network as a graph, by identifying maximal $k$-trusses, we can identify a set of cohesive groups or communities. The granularity can be adjusted through the setting of $k$. With the identified communities, analysis can be done to extract local information from the communities and follow-up commercial strategies, such as personalized advertising~\cite{DBLP:journals/ijcsa/MeyerBOS11}, viral marketing~\cite{RePEc:ucp:jconrs:v:14:y:1987:i:3:p:350-62}, can then be devised based on collected information. In fact, the problem of finding cohesive subgraphs has been studied for a long time. There are many other models to identify cohesive subgraphs, e.g., clique, n-clique, quasi-clique, n-clan, n-club, k-plex and k-core. Despite many such models have been proposed, to the best of our knowledge, very little attention has been paid to maintaining the cohesive subgraphs in evolving networks. However, real networks are updated frequently. %It is recorded that 700,000 people join Facebook per day~\cite{Facebook}. For LinkedIn, the membership grows by approximately two new members every second~\cite{LinkedIn}. 
Every second, there may be people joining and leaving a network, setting up or eliminating links to others. For example, the membership of LinkedIn grows by approximately two new members every second~\cite{LinkedIn}. As a result, it is imperative to maintain the discovered interesting communities (cohesive subgraphs) dynamically if possible rather than to recompute the communities from scratch each time to respond to the updating of nodes and edges of the network. In this paper, we focus on dynamically maintaining maximal trusses in evolving networks.

We first introduce other forms of cohesive subgraphs and then propose the challenge of maintaining trusses. 

The basic notion is clique~\cite{RePEc:spr:psycho:v:14:y:1949:i:2:p:95-116}. A clique is a subgraph where every node has links to other nodes in the subgraph. The definition of clique is too rigid. There are relaxed forms of cliques, such as quasi-clique~\cite{DBLP:conf/latin/AbelloRS02, DBLP:journals/tods/ZengWZK07} and n-clique~\cite{n_clique}. Quasi-clique relaxes the definition of clique by asking each node to connect to a certain percentage of other nodes. Quasi-clique can be defined based nodes~\cite{DBLP:journals/tods/ZengWZK07} and edges~\cite{DBLP:conf/latin/AbelloRS02}. $n$-clique relaxes the distance between any two nodes in a clique from 1 to $n$, i.e. two nodes can be considered as connected if they are connected by at most $n$ hops of edges. $n$-clan~\cite{n_clan} is a restricted form of $n$-clique by further imposing a constraint on the diameter. $n$-club~\cite{n_clan} removes the $n$-clique requirement from the $n$-clan, asking for only diameter constraint. $k$-plex~\cite{Seidman1978} is also a relaxed form of clique. It requires each node in a clique of $x$ nodes to have $x-k$ links to other nodes. Here, $k$ can be considered as a tolerant factor. $k$-plex is similar to quasi-clique. The difference is that quasi-clique asks for a percentage but $k$-plex asks for a number of links. The common feature of the above notions is that their computation are all NP-hard. 

Another interesting notion, recently attracting the attention of researchers, is $k$-core~\cite{Seidman1983269}. $k$-core is the largest subgraph where each node has its degree at least $k$.  $k$-core can somehow catch a cohesive graph, because in such a graph, every node has at least $k$ links to other nodes. The computation and maintenance of $k$-core is polynomial. However, $k$-core may not be very accurate because the existence of bridges. %For example, in Fig.~\ref{fig:intro_example}, it is more suitable to regard the graph as two cohesive subgraphs, though the graph itself is a $k$-core. 
This has been observed in~\cite{DBLP:conf/edbt/ZhouLYLCL12}, where the authors use maximal induced k-connected subgraphs (or $k$-strong subgraphs termed in~\cite{DBLP:journals/corr/cs-DS-0207078}) to model cohesive subgraphs. The computation of k-strong subgraphs is also expensive, though polynomial. The maintenance of k-strong subgraphs is also time-consuming~\cite{DBLP:journals/tpds/LiangBS01}. 

Considering the computation and maintenance difficulty of $k$-strong subgraphs, we are expecting a model that should be less expensive to compute and maintain, but still preserve the good property of $k$-strong subgraphs. $k$-truss fits in this category. Firstly, a $k$-truss is both a $(k-1)$-core~\cite{Wang:2012:TDM:2311906.2311909} and a $(k-1)$-connected subgraph, since a $(k-1)$-core is a $(k-1)$-connected graph. Secondly, the computation of $k$-truss is $O(m^{1.5})$, where $m$ is the number of edges. It is low polynomial. However, it is still infeasible to recompute the trusses frequently on evolving networks. To tackle this problem, we study how to maintain $k$-trusses in this paper. Our result is that maintaining $k$-truss is $O(|E_l|)$, where $E_l$ is the set of affected edges whose truss numbers need to be updated. In a large graph, we usually have $|E_l| << |E|$. This means $k$-truss may be a promising structure to model cohesive subgraphs in practice.

\begin{example}
Fig.~\ref{fig:truss1} shows the relationships of actors studied by S. Freeman and L. Freeman~\cite{DBLP:journals/web/truss}. Fig.~\ref{fig:truss1} shows the entire network. Fig.~\ref{fig:truss2} shows the maximal 4-trusses. There are two in total. As the example shows, in either 4-truss, each edge is in at least 2 triangles. Fig.~\ref{fig:truss3} shows the 3-core of the network, where each node has at least 3 neighbors. The 3-strong subgraph of the network is the same as the 3-core in this example.  
\end{example}

\begin{figure}[t]
  \centering
  \subfigure[Reciprocated Friendship and Meeting Relations]{\label{fig:truss1}
    \includegraphics[height=64mm,width=80mm]{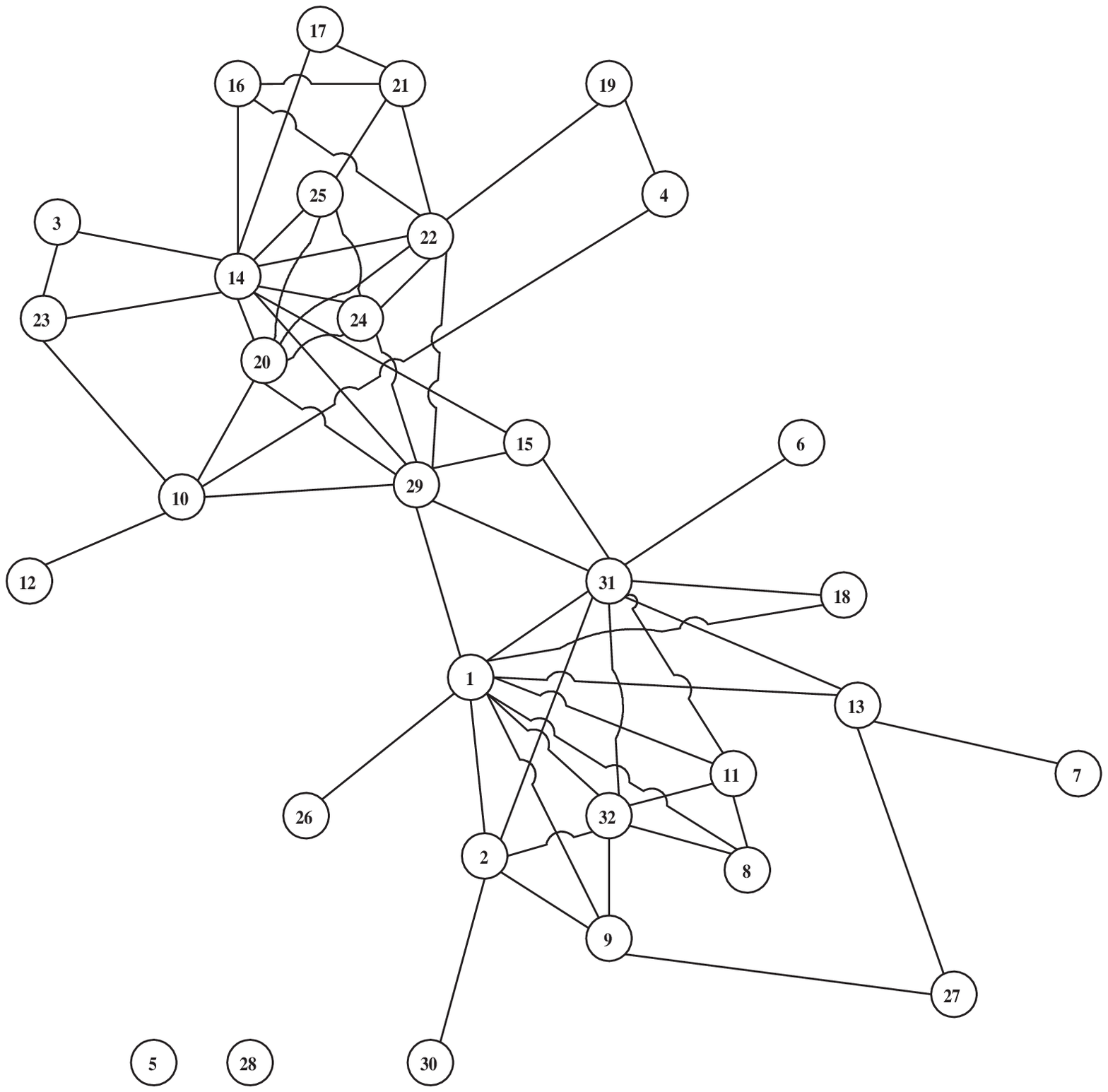}}
 %\hskip 0.1in
  \subfigure[4-trusses]{\label{fig:truss2}
    \includegraphics[height=45mm,width=40mm]{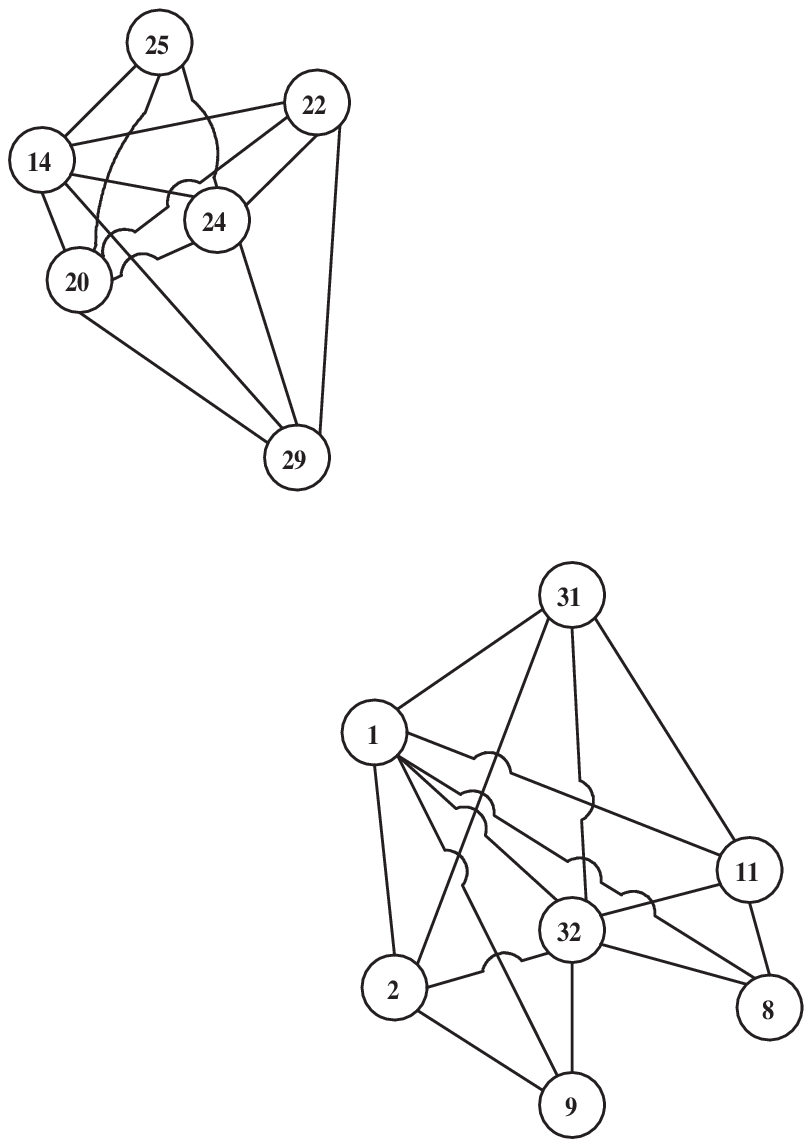}} 
 %\hskip 0.1in
  \subfigure[3-cores]{\label{fig:truss3}
    \includegraphics[height=45mm,width=40mm]{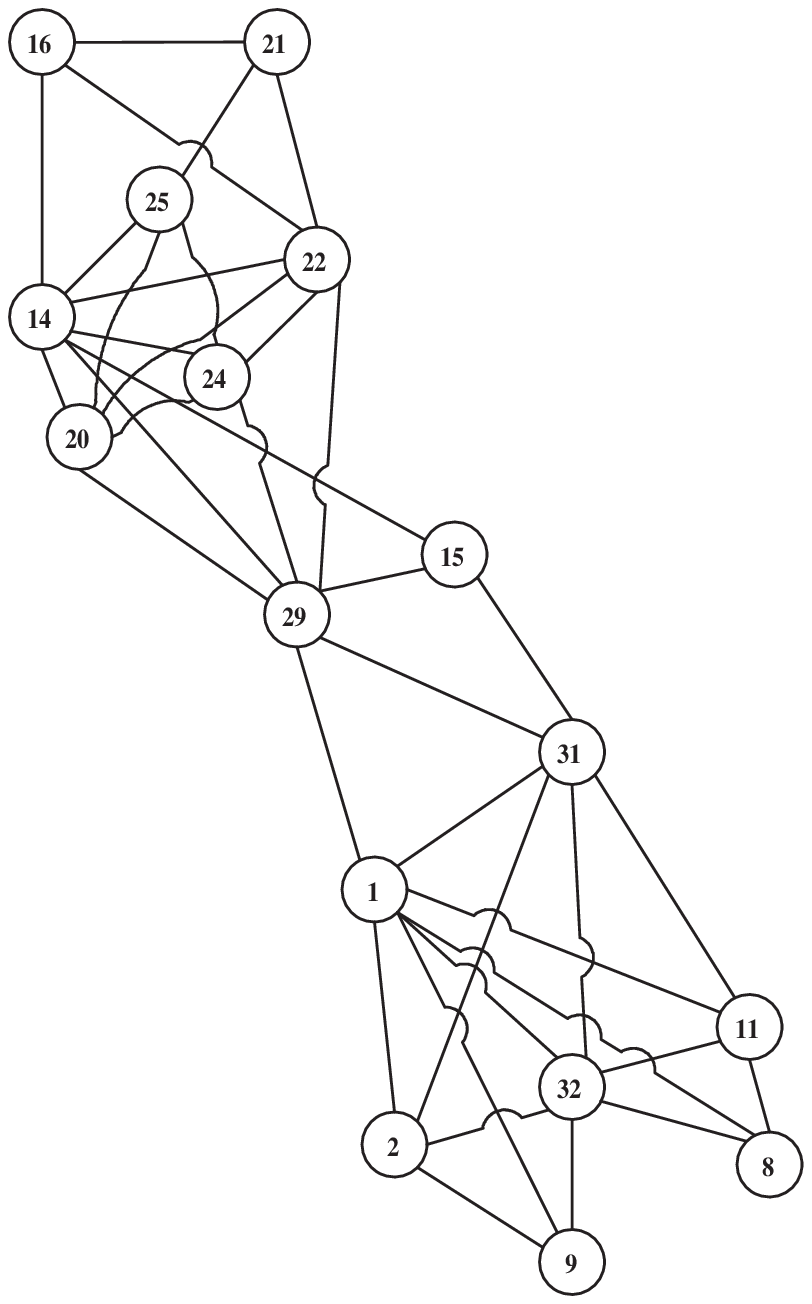}}
\caption{Examples of $k$-trusses and $k$-cores}
\label{fig:truss}
\end{figure} 
  
Finally, we summarize the main contributions of this paper as follows:
\begin{itemize}
    \item We are the first to study the truss maintenance problem.

    \item We discuss a set of important truss maintenance properties by answering the following questions ``which trusses may be affected'', ``how much effect will an update produce'', ``how far will the effect spread''. 
    
    \item We devise effective truss maintenance algorithms for edge deletions and edge insertions. We prove the correctness of the algorithms and analyze the complexity of the maintenance algorithms.
    
    \item We propose to build truss indexes to speed up the evaluation of truss queries. We also introduce how to maintain the indexes.

    \item The truss maintenance algorithms and truss index maintenance algorithms are demonstrated as efficient by the experiments.
\end{itemize}

Here is a roadmap of this paper. In Section~\ref{sect:prel}, we provide a formal definition of the problem, and introduce some necessary notations. In Section~\ref{sect:properties}, we discuss the properties of truss maintenance. In Section~\ref{sect:algorithms}, we show the algorithms of performing edge deletions or insertions. In Section~\ref{sect:index}, we show how to build and maintain truss indexes. Experiment results are shown in Section~\ref{sect:experiments}. Related works and conclusions are in Section~\ref{sect:related} and Section~\ref{sect:conclusion} respectively.

\section{Preliminaries}
\label{sect:prel}

In this work, we consider only undirected, unweighted simple graphs. The truss definition and maximal truss definition have already been introduced in Section 1. Apparently, if an edge $e$ is in a $k$-truss, it must be in a $(k-1)$-truss which is a supergraph of the $k$-truss. We define $\phi(e)$ as the maximal truss number of the trusses that edge $e$ can reside in. This truss is denoted as $T_{\phi(e)}$. We define the global support of an edge $e$ as the number of triangles containing the edge in the graph $G=(V,E)$, denoted as $sup(e,G)$. Similarly, we define the local support of an edge $e$ in a subgraph $G_s=(V_s,E_s)$ ($V_s \subseteq V$ and $E_s \subseteq E$) as the number of triangles containing $e$ in $G_s$, denoted as $sup(e,G_s)$.
%For a given $k$, $k$-truss is unique in a graph.

%$k$-truss and maximal $k$-truss are defined in the introduction. Let $\phi(e)$ denote the maximal truss number that $e$ is contained in a truss.

\begin{lemma}
Given a graph $G$, we have the following:
\[
\phi(e) \leq sup(e,T_{\phi(e)}) + 2 \leq sup(e,G) + 2
\]

\label{lemm:support}
\end{lemma}

%For an edge $e$, we have the property: the truss number of $e$ <= local support of $e$ <= global support of $e$. 

To help with the discussions, we introduce some notions. Let $a$, $b$ be two nodes, $e=(a,b)$ be an edge to be deleted or inserted between nodes $a$ and $b$, $n(a)$ and $n(b)$ be the neighbor nodes of $a$ and $b$ respectively, we define $S_{a,b} = n(a) \cap n(b)$ as the common neighbors of $a$ and $b$, define $E_{S_{a,b}\leftrightarrow \{a,b\}} = \{(n_1,n_2): n_1 \in S_{a,b}, n_2 \in \{a,b\} \}$ as the set of edges between $S_{a,b}$ and $\{a,b\}$. Furthermore, we define $k_{max}(e)$ and $k_{min}(e)$ as the maximum and minimum truss number of the edges in  $E_{S_{a,b}\leftrightarrow \{a,b\}}$, i.e. $k_{max}(e) = max\{\phi(e') : e' \in E_{S_{a,b}\leftrightarrow \{a,b\}}  \}$ and $k_{min}(e) = min\{\phi(e') : e' \in E_{S\leftrightarrow \{a,b\}} \}$. 

%Property: k-truss may not contain a k-clique

%Define support and local support. $E_{T_k}$ as $k$-truss of $G$, truss number $\phi(e)$.

%\begin{lemma}
%The local support of an edge is no less than the truss number of the edge.
%\label{lemma:localsupport}
%\end{lemma}

A notation table is given in Table~\ref{tab:notations} for easy reference. In the main part of our discussion, we consider one insertion and one deletion only, multiple insertions and deletions can be regarded as repeating single insertion and deletion.

\begin{table}[t]
\small\small \label{tab:notations} \caption{Notations and Descriptions}

\begin{center}

{\footnotesize

\begin{tabular}{l|l}

\hline
Notation &  Description\\
\hline \hline
$\phi(e)$ & the maximal truss number of the trusses\\
          & that edge $e$ resides in\\
\hline
$T_{\phi(e)}$ & the maximal $\phi(e)$-truss that \\
              & contains edge $e$ \\
\hline
$n(a)$      & the set of neighbor nodes of node $a$ \\
\hline
$S_{a,b}$ & the common neighbors of node $a$ and $b$: \\
& $n(a) \cap n(b)$ \\
\hline
$E_{S_{a,b}\leftrightarrow \{a,b\}}$ &  the set of edges between $S_{a,b}$ and $\{a,b\}$: \\ 

& $\{(n_1,n_2): n_1 \in S_{a,b}, n_2 \in \{a,b\} \}$ \\
\hline
$k_{max}(e)$ & the maximal truss number of the edges in \\
& $E_{S_{a,b}\leftrightarrow \{a,b\}}$: $max\{\phi(e') : e' \in E_{S_{a,b}\leftrightarrow \{a,b\}}  \}$ \\
\hline
$k_{min}(e)$ & the minimal truss number of the edges in \\
& $E_{S_{a,b}\leftrightarrow \{a,b\}}$: $min\{\phi(e') : e' \in E_{S_{a,b}\leftrightarrow \{a,b\}}  \}$ \\
\hline
\end{tabular}
}

\end{center}
\end{table}

%For short, we use insertion and deletion to refer edge insertion and edge deletion. 
%In the main part of our discussion, we consider one insertion and one deletion only, multiple insertions and deletions can be regarded as repeating single insertion and deletion. %In the experiment part, we will show the ``competence ratio'' of our methods.

\section{Truss Maintenance Properties}
\label{sect:properties}

In this section, we discuss some properties for truss maintenance under one insertion and one deletion. Before we proceed introducing the algorithms in Section~\ref{sect:algorithms} to tackle the truss maintenance problem, we realize it is important and helpful to discuss some theoretical findings first. This will underpin the design of the algorithms. Given a deletion or an insertion, we believe a reader is prone to ask the following questions:

\begin{enumerate}
 \item[$Q_1$] \emph{Which} trusses may be affected? -- A set of $k$-trusses with different $k$'s may be affected. What are the affected $k$'s?
 \item[$Q_2$] \emph{How much} effect will an update produce? -- It is natural to assume that one update will affect the truss numbers of the other edges by at most 1. Is this true?
 \item[$Q_3$] \emph{How far} will the effect spread? -- Will the effect be bounded within an area or can the effect spread far away?
\end{enumerate}

Before answering the questions raised above, we introduce the following observation. The correctness is obvious.  

\begin{observation}
After an insertion, the truss number of any edge will not decrease; after a deletion, the truss number of any edge will not increase.
\label{lemma:obvious}
\end{observation}

%Throughout this paper, for the sake of easy understanding, we start discussions on deletions first and then insertions. When we say ``an edge is affected'', we mean ``the truss number of an edge is affected''.

\subsection{How much effect will an update produce?}

We start with the ``\emph{how much}'' question ($Q_2$) first, because the answer is easier to understand and will be used to answer the ``\emph{which}'' question ($Q_1$). The answer to $Q_2$ is given in Lemma~\ref{lemma:how_much_effect}.
 
\begin{lemma}
After one deletion or insertion, if an edge is affected, the truss number is affected by at most 1, i.e. increased by 1 or decreased by 1.
\label{lemma:how_much_effect}
\end{lemma}

\begin{proof}
We prove the deletion first and then the insertion.

Let $e$ be an edge to be deleted and its truss number be $\phi(e)$, firstly any edge $e'$ with truss number $\phi(e')$ more than $\phi(e)$ is not affected, because $e$ is not in the $\phi(e')$-truss since $\phi(e')>\phi(e)$, therefore deleting $e$ will not affect any edge in the $\phi(e')$-truss, i.e. $e'$ is still in the $\phi(e')$-truss. Secondly, this time, let $e'$ be any edge in the $\phi(e')$-truss $T_{\phi(e')}$ with $\phi(e') \leq \phi(e)$, after deleting $e$, the global and local support of $e'$ will reduce by at most 1 respectively (both may remain unchanged), because $e$ and $e'$ only participate together in at most one triangle. As a result, after deleting $e$, $e'$ is guaranteed to be in the $(\phi(e')-1)$-truss. In both cases, the effect is not more than 1.

The insertion follows the deletion, and can be proved by contradiction. Suppose inserting $e$ increases the truss number of an edge $e'$ by more than 1, then for the updated graph, deleting $e$ will cause the truss number of $e'$ to decrease by more than 1. This contradicts the first (deletion) part of the proof.  
\end{proof}

\subsection{Which trusses may be affected?}

The answer to this ``\emph{which}'' question is given in Theorem~\ref{theo:which_affected_delete} (for deletion) and Theorem~\ref{theo:which_affected_insert} (for insertion). 

\begin{theorem}
For deletion, let $e=(a,b)$ be the to-be-deleted edge, (a) if $S_{a,b} = \emptyset$ or $k_{min}(e) > \phi(e)$, no other edges are affected; (b) if $S_{a,b} \neq \emptyset$ and $k_{min}(e) \leq \phi(e)$, the possibly affected trusses must have truss numbers within the range $[k_{min}(e), \phi(e)]$.
\label{theo:which_affected_delete}
\end{theorem}

\begin{proof} For case (a), if $S_{a,b} = \emptyset$, it is obvious that no other edges are affected, because no triangles contain $(a,b)$; if $k_{min}(e) > \phi(e)$, no other edges are affected, because any edge $e'$ with truss number $\phi(e')$ more than $\phi(e)$ is not affected. The reason is that $e=(a,b)$ is not in the $\phi(e')$-truss containing $e'$ since $\phi(e')> \phi(e)$, therefore deleting $e$ will not affect $e'$.

For case (b), $S_{a,b} \neq \emptyset$ and $k_{min}(e) \leq \phi(e)$, (1) firstly, again, as proved in case (a), any edge $e'$ with truss number $\phi(e')$ more than $\phi(e)$ is not affected. (2) next, we prove that any edge with truss number less than $k_{min}(e)$ will not be affected. Let $e'$ be an edge with truss number $\phi(e') < k_{min}(e)$, then according to Lemma~\ref{lemma:how_much_effect}, after deleting $e$, all edges in $E_{S_{a,b}\leftrightarrow \{a,b\}}$ remain in the $\phi(e')$-truss, because the truss number of an edge can be affected by at most 1 and the minimum truss number directly affected by deleting $e$ is $k_{min}(e)$. As a result, the updated $\phi(e')$-truss is the same as the old $\phi(e')$-truss but without the deleted edge $e$, and thus the edge $e'$ remains in the $\phi(e')$-truss. As a result, combining case (1) and case (2), the possibly affected trusses are in the range of $[k_{min}(e), \phi(e)]$.
\end{proof}

Moreover, we show that the range $[k_{min}(e), \phi(e)]$ is tight with an example. Let $K_n$ denote a complete graph with $n$ nodes, suppose a graph is joined by a set of complete graphs $K_{k}, K_{k-1},...,K_{t}$ $(t \geq 3)$ on an edge $e$, i.e. $e$ is the only shared edge between $K_{k}, K_{k-1},...,K_{t}$, in such a graph, all edges in the subgraph $K_k$ including $e$ have the truss number $k = \phi(e)$, all edges in the subgraph $K_{k-1}$ except $e$ have the truss number $k-1$, ..., all edges in the subgraph $K_t$ except $e$ have the truss number $t$. Apparently, in this example, $k_{min}(e)=t$. If we delete the edge $e$, we can find that the truss numbers of all other edges will decrease by 1. These edges originally have the truss numbers ranging from $t=k_{min}(e)$ to $k=\phi(e)$, therefore the range $[k_{min}(e), \phi(e)]$ is tight. Fig.~\ref{fig:tightness} is one such example. A $K_3$, a $K_4$ and a $K_5$ are joined at edge $(a,b)$. If we delete $(a,b)$, all the other edges are affected and they originally have the truss range $[3,5]$. Here, for the edge $(a,b)$, $k_{min}((a,b))=3$ and $\phi((a,b))=5$. The affected range $[3,5]$ is tight in this example.

\begin{figure}[t]
    \centering
    {\includegraphics[width=60mm]{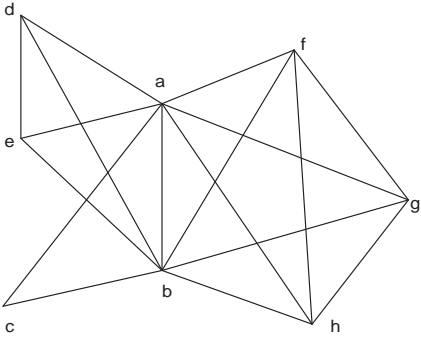}}
    \caption{A $K_3$, a $K_4$ and a $K_5$ are joined at $(a,b)$}
    \label{fig:tightness}
\end{figure}

Before presenting Theorem~\ref{theo:which_affected_insert} for the insertion, we introduce Lemma~\ref{lemma:together} and Lemma~\ref{lemma:truss_bound} first to support the proof of Theorem~\ref{theo:which_affected_insert}. 

%To help with the discussion, we introduce some notions first. Let $a$, $b$ be two nodes, $e=(a,b)$ be an edge to be deleted or inserted between nodes $a$ and $b$, $nb(a)$ and $nb(b)$ be the neighbor nodes of $a$ and $b$ respectively, let $S = nb(a) \cap nb(b)$ be the common neighbors of $a$ and $b$, let $E_{S\leftrightarrow \{a,b\}} = \{(n_1,n_2): n_1 \in S, n_2 \in \{a,b\} \}$ be the set of edges between $S$ and $\{a,b\}$, let $k_{max}$ and $k_{min}$ be the maximum and minimum truss number of the edges in  $E_{S\leftrightarrow \{a,b\}}$, i.e. $k_{max} = max\{\phi(e) : e \in E_{S\leftrightarrow \{a,b\}}  \}$ and $k_{min} = min\{\phi(e) : e \in E_{S\leftrightarrow \{a,b\}} \}$. We have the main result as Theorem~\ref{theo:which_affected_delete} and~\ref{theo:which_affected_insert}. %A EXAMPLE HERE TO EXPLAIN THE NOTIONS.
%Before introducing Theorem~\ref{theo:which_affected_delete} and~\ref{theo:which_affected_insert}, we give Lemma~\ref{lemma:together} and~\ref{lemma:truss_bound} first in order to support the proof of the theorem. 

\begin{lemma}
After inserting an edge $e$, if the truss number of an edge $e'$ increases from $\phi(e')$ to $\phi(e')+1$, then the inserted edge $e$ must be in the $(\phi(e')+1)$-truss containing $e'$ in the updated graph.
\label{lemma:together}
\end{lemma}

\begin{proof}
By contradiction, if the inserted edge $e$ is not in the $(\phi(e')+1)$-truss containing $e'$ in the updated graph, then deleting $e$ will not reduce the truss number of $e'$. This contradicts with the assumption that after inserting $e$, the truss number of $e'$ increases from $\phi(e')$ to $\phi(e')+1$.
\end{proof}

%\begin{lemma}
%In a k-truss subgraph, there must be a node with local triangles $(k-2)$, otherwise the $k$-truss will be $(k+1)$-truss according to definition.
%\end{lemma}

\begin{lemma}
For any edge $e=(a,b)$ in a graph, we have the following: %$k \leq max\{\phi(e): e \in E_{S\leftrightarrow \{a,b\}}  \} $. %We can let the property hold for $k=2$ by defining the highest truss number as 2 when $S=\emptyset$.
\[
%\begin{equation}
 \left\{ {\begin{array}{*{20}c}
   {\phi(e) =2 } & {(E_{S_{a,b}\leftrightarrow \{a,b\}} = \emptyset)}  \\
   {\phi(e) \leq max\{\phi(e'): e' \in E_{S_{a,b}\leftrightarrow \{a,b\}}  \} } & {(E_{S_{a,b}\leftrightarrow \{a,b\}} \neq \emptyset)}  \\
\end{array}} \right.
\]
%\label{eq:truss_bound}
%\end{equation}
\label{lemma:truss_bound}
\end{lemma}

\begin{proof} If $E_{S_{a,b}\leftrightarrow \{a,b\}} = \emptyset$ (i.e. $S_{a,b}=n(a) \cap n(b)=\emptyset$), the truss number of $(a,b)$ is 2. The lemma holds naturally. If $E_{S_{a,b}\leftrightarrow \{a,b\}} \neq \emptyset$, assume $\phi(e) > max\{\phi(e'): e' \in E_{S_{a,b}\leftrightarrow \{a,b\}}  \} $, then we have that edge $(a,b)$ is in a $\phi(e)$-truss and all the edges in $E_{S_{a,b}\leftrightarrow \{a,b\}}$ are not in the $\phi(e)$-truss. Therefore, edge $(a,b)$ does not have supporting triangles in the $\phi(e)$-truss, so $\phi(e)$ must be 2. This contradicts with the assumption $\phi(e) > max\{\phi(e'): e' \in E_{S_{a,b}\leftrightarrow \{a,b\}}  \} $ in which it is apparent that $max\{\phi(e'): e' \in E_{S_{a,b}\leftrightarrow \{a,b\}}  \} \geq 3$ if $E_{S_{a,b}\leftrightarrow \{a,b\}} \neq \emptyset$, so $\phi(e) \leq max\{\phi(e'): e' \in E_{S_{a,b}\leftrightarrow \{a,b\}}  \} $ when $E_{S_{a,b}\leftrightarrow \{a,b\}} \neq \emptyset$. 
\end{proof}

Lemma~\ref{lemma:truss_bound} says the maximum truss number of edges in $E_{S_{a,b}\leftrightarrow \{a,b\}}$ is an upper bound of the truss number of $(a,b)$. Similarly, readers may wonder ``do we have the minimum truss number of the edges in $E_{S_{a,b}\leftrightarrow \{a,b\}}$ as a lower bound of the truss number of $(a,b)$ when $E_{S_{a,b}\leftrightarrow \{a,b\}} \neq \emptyset$?'', i.e. ``do we have $\phi(e) \geq min\{\phi(e'): e' \in E_{S_{a,b}\leftrightarrow \{a,b\}}  \} $ when $E_{S_{a,b}\leftrightarrow \{a,b\}} \neq \emptyset$?'' The answer is no. We illustrate this through an example in Fig.~\ref{fig:counter_example}, where $\phi((a,b))=3$, $S_{a,b}=n(a)\cap n(b)=\{c\}$, $E_{S_{a,b}\leftrightarrow \{a,b\}} = \{(a,c),(b,c)\}$, $min\{\phi(e'): e' \in E_{S_{a,b}\leftrightarrow \{a,b\}}  \}$ = $min\{\phi((a,c))=4,\phi((b,c))=4 \} = 4$. In this example, $\phi((a,b)) < min\{\phi(e'): e' \in E_{S_{a,b}\leftrightarrow \{a,b\}}  \}$.

\begin{figure}[t]
    \centering
    {\includegraphics[width=60mm]{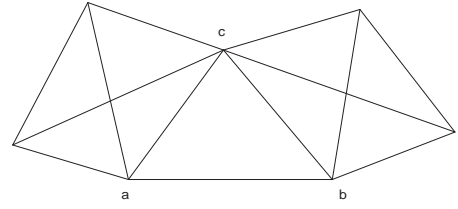}}
    \caption{An example showing $\phi((a,b)) \geq min\{\phi(e'): e' \in E_{S_{a,b}\leftrightarrow \{a,b\}}  \} $ does not hold. }
    \label{fig:counter_example}
\end{figure}

%The lemma can be proved using an example. Let (a,b) be an edge between two k-trusses (k >= 3), then the truss number of (a,b) is 2, but (a,b) neighbour edges all have truss number 3. 

Now we are ready to state the following theorem. 
%introduce and prove Theorem~\ref{theo:which_affected_insert}.

\begin{theorem}
For insertion, let $e=(a,b)$ be the newly inserted edge, (a) if $S_{a,b} = \emptyset$ or $k_{min}(e) > |S_{a,b}|+1$, no other edges are affected; (b) if $S_{a,b} \neq \emptyset$ and $k_{min}(e) \leq |S_{a,b}|+1$, the possibly affected trusses must have truss numbers within the range $[k_{min}(e), min(|S_{a,b}|+1, k_{max}(e))]$.
%\[
%\begin{equation}
% \left\{ {\begin{array}{*{20}c}
%   {[k_{min}(e), min(|S_{a,b}|+1, k_{max}(e))]} & {( k_{min}(e) \leq |S_{a,b}|+1 )}  \\
%   { \emptyset } & {(k_{min}(e) > |S_{a,b}|+1)}  \\
%\end{array}} \right.
%\]
\label{theo:which_affected_insert}
\end{theorem}

\begin{proof} For case (a), if $S_{a,b} = \emptyset$, obviously no other edges are affected; if $k_{min}(e) > |S_{a,b}|+1$, no other edges are affected because any edge $e'$ with $\phi(e') > |S_{a,b}|+1$ is not affected. Assume the truss of $e'$ is affected and increased to $\phi(e')+1$, according to Lemma~\ref{lemma:together}, $(a,b)$ will be in the $(\phi(e')+1)-truss$ with the truss number larger than $|S_{a,b}|+2$. However, it is obvious that $|S_{a,b}|$ is the global support of $(a,b)$, so $|S_{a,b}|+2$ is an upper bound of $(a,b)$'s truss number. A contradiction appears.  

For case (b), $S_{a,b} \neq \emptyset$ and $k_{min}(e) \leq |S_{a,b}|+1$, if an edge $e'$ is affected, then we must have (1) $\phi(e') \leq |S_{a,b}|+1$, (2) $\phi(e') \leq k_{max}(e)$ and (3) $\phi(e') \geq k_{min}(e)$. We now prove each part respectively. (1) Firstly, $\phi(e') \leq |S_{a,b}|+1$ is proved in case (a). (2) Secondly, any edge $e'$ with $\phi(e') > k_{max}(e)$ is not affected. We prove by contradiction. Assume the truss number of $e'$ is increased to $\phi(e')+1$, according to Lemma~\ref{lemma:together}, the truss number of the inserted edge $e=(a,b)$, $\phi(e)$, is at least $\phi(e')+1$ in the updated graph, so we have $\phi(e) > k_{max}(e) + 1$. On the other hand, according to Lemma~\ref{lemma:how_much_effect}, in the updated graph, $max\{\phi(e') : e' \in E_{S_{a,b}\leftrightarrow \{a,b\}}  \} \leq k_{max}(e) + 1$. Furthermore, according to Lemma~\ref{lemma:truss_bound}, in the updated graph, $\phi(e) \leq max\{\phi(e') : e' \in E_{S_{a,b}\leftrightarrow \{a,b\}}  \} \leq k_{max}(e) + 1$. There is a contradiction. (3) Finally, any edge $e'$ with $\phi(e') < k_{min}(e)$ is not affected. Since $\phi(e') < k_{min}(e)$, all the edges in $E_{S_{a,b}\leftrightarrow \{a,b\}}$ are in a $(\phi(e')+1)$-truss. As a result, there is no new edges adding into the $(\phi(e')+1)$-truss except the newly inserted one, because all the directly affected edges $E_{S_{a,b}\leftrightarrow \{a,b\}}$ are all already in the $(\phi(e')+1)$-truss. The unchangingness of any $(\phi(e')+1)$-truss with $\phi(e')<k_{min}(e)$, together with Observation~\ref{lemma:obvious}, implies that current $\phi(e')$-truss edges remain as $\phi(e')$-truss edges. 
\end{proof}

Then, let us inspect the tightness of $[k_{min}(e), min\{|S_{a,b}|+1, k_{max}(e)\}]$. Similar to the deletion, suppose a graph is joined by a set of complete graphs $K_{k}, K_{k-1},...,K_{t} (t \geq 3)$ on an edge $e$, but now $e$ was deleted, and we want to insert $e$ back into the $e$-deleted graph. Then the affected trusses will have truss ranges from $t-1$ to $k-1$. $t-1$ and $k-1$ are the $k_{min}(e)$ and $k_{max}(e)$ of the to-be-inserted edge $e$ in the $e$-deleted graph. Sometimes, $k_{max}(e)$ may be larger than $|S_{a,b}|+1$. In such cases, $[k_{min}(e),|S_{a,b}|+1]$ is the affected truss range. See Fig.~\ref{fig:k_max_large} for an example, $(a,b)$ is the edge to be inserted, $k_{max}((a,b))$ is 4 on the edge $(b,c)$, $k_{min}((a,b))$ is 3 on the edges $(a,d)$, $(b,d)$ and $(a,c)$, $S_{a,b}=\{c,d\}$, so $k_{min}((a,b)) = |S_{a,b}|+1 = 3 < k_{max}((a,b))$. The affected truss range is $[k_{min}((a,b)), |S_{a,b}|+1]$, which is $[3,3]$. In this example, the affected edges are $(a,d)$, $(b,d)$, $(a,c)$ and $(c,d)$. All of them originally have truss number 3. In the next example, $k_{min}((a,b)) > |S_{a,b}|+1$, see Fig.~\ref{fig:k_min_large}. Again, $(a,b)$ is the to-be-inserted edge, $S_{a,b}=\{c\}$, $k_{max}((a,b))=4$ on the edge $(b,c)$, $k_{min}((a,b))=3$ on the edge $(a,c)$, $|S_{a,b}|+1 = 2 < k_{min}((a,b))$. No existing edges are affected when we insert $(a,b)$.

\begin{figure}[t]
    \centering
    {\includegraphics[width=60mm]{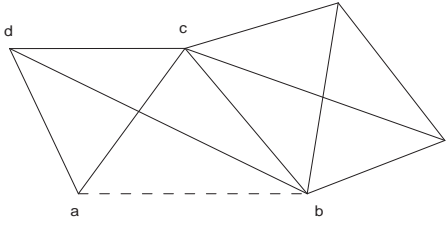}}
    \caption{$k_{max}((a,b))=4$, $k_{min}((a,b))=3$, $S_{a,b}=\{c,d\}$, $k_{min}((a,b)) \leq |S_{a,b}|+1<k_{max}((a,b))$, so the affected truss range is $[k_{min}((a,b)), |S_{a,b}|+1]$.}
    \label{fig:k_max_large}
\end{figure}

\begin{figure}[t]
    \centering
    {\includegraphics[width=60mm]{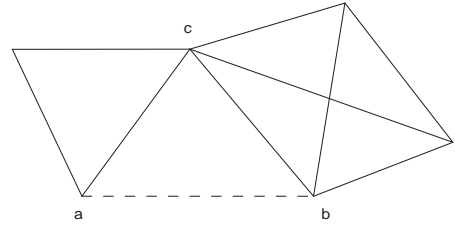}}
    \caption{$k_{max}((a,b))=4$, $k_{min}((a,b))=3$, $S_{a,b}=\{c\}$, $k_{min}((a,b)) > |S_{a,b}|+1$, so no other edges are affected.}
    \label{fig:k_min_large}
\end{figure}

Theorem~\ref{theo:which_affected_delete} and~\ref{theo:which_affected_insert} tell us that, for an update, we do not need to inspect all the edges in the graph, instead we only need to inspect a subset of edges whose truss numbers are within a certain range.

\subsection{How far will the effect spread?}
\label{sect:how_far}

From the above discussions, we find that the truss number of an edge is affected because the local support of the edge is changed. When we delete or insert an edge $(a,b)$, recall that $S=n(a) \cap n(b)$, the set of edges that are directly affected is $E_{S_{a,b}\leftrightarrow \{a,b\}}$, because $(a,b)$ forms triangles with the edges in $E_{S_{a,b}\leftrightarrow \{a,b\}}$. As a result, the global support of the edges in $E_{S_{a,b}\leftrightarrow \{a,b\}}$ will change, and this may lead to changes of their local support, eventually the truss numbers of these edges may change. Furthermore, the effect may spread to the neighbor edges of the edges in $E_{S_{a,b}\leftrightarrow \{a,b\}}$. In this section, we will discuss how far the effect will spread. For simplicity, we only discuss deletion in all the examples in this section. The effect of insertion is similar. (one can simply insert the deleted edge back and see the impact of an insertion in the following examples.)

Firstly, we show that the effect may spread indeed. It is natural to recognize that the truss numbers of edges in $E_{S_{a,b}\leftrightarrow \{a,b\}}$ may be affected, but could other edges be affected, such as the edges incident on node $a$ or $b$ but not in $E_{S_{a,b}\leftrightarrow \{a,b\}}$? The answer is yes. We give an example in Fig.~\ref{fig:how_far_one_step}. The graph is a 4-truss. Every edge is in two triangles. If we delete $(a,b)$, directly we will have $(a,c)$, $(a,d)$, $(b,c)$, $(b,d)$ having global support 1, and thus these four edges cannot be in a 4-truss. Further, we find that there is no 4-truss in the updated graph. The truss numbers of all the other edges are reduced from 4 to 3, such as $(a,e)$, $(b,f)$, though they do not belong to $E_{S_{a,b}\leftrightarrow \{a,b\}}$. 

\begin{figure}[t]
    \centering
    {\includegraphics[width=48mm]{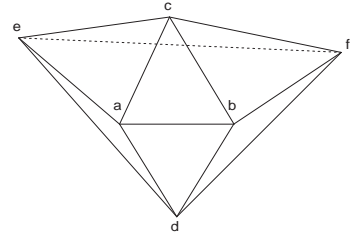}}
    \caption{The effect may spread}
    \label{fig:how_far_one_step}
\end{figure}

Next, we give another example to show that the effect can be spread a long way away (summarized as Lemma~\ref{lemm:unlimited}). In Fig.~\ref{fig:how_far_unlimited}, there is an unlimited graph piling by the unit shown in the bottom right corner. In such a graph, every edge is in two triangles and every edge has truss number 4. If we delete edge $(a,b)$, the first set of edges whose truss numbers are affected is $(a,p_2)$, $(b,p_2)$, $(a,q_3)$, $(b,q_3)$. Each of these four edges will have only one supporting triangle, so they will no longer be in the 4-truss. Then, the effect will spread towards four directions leading to the truss number of every edge decreased by 1. 

\begin{lemma}
There is no guarantee that the update effect will end after spreading a fixed number of steps.  
\label{lemm:unlimited}
\end{lemma}

\begin{figure}[t]
    \centering
    {\includegraphics[width=63mm]{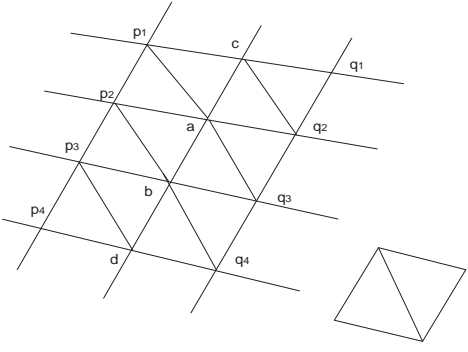}}
    \caption{The effect may spread very far}
    \label{fig:how_far_unlimited}
\end{figure}

\section{Algorithms}
\label{sect:algorithms}

Based on the properties discussed in Section~\ref{sect:properties}, we design maintenance algorithms for deleting an edge and inserting an edge.

\subsection{Deleting an edge}
\label{subsect:delete}

Inspired by the answers to the questions ``\emph{which trusses may be affected}'' and ``\emph{how far will the effect spread}'' in Section~\ref{sect:properties}, we can devise an outward inspection based truss maintenance approach for edge deletion. Firstly, we start from the to-be-deleted edge $(a,b)$ and check which neighbor edges will be directly affected if we delete $(a,b)$. These directly affected edges must be in the edge set $E_{S_{a,b}\leftrightarrow \{a,b\}}$, because other edges (the edges not in $E_{S_{a,b}\leftrightarrow \{a,b\}}$) do not participate in triangles together with $(a,b)$, and thus they are definitely not affected at the moment, although we may find that they are affected at a later time. After inspecting the set $E_{S_{a,b}\leftrightarrow \{a,b\}}$, we may find some edges in $E_{S_{a,b}\leftrightarrow \{a,b\}}$ are affected while the others are not. Then, we continue inspecting the neighbor edges of the affected edges in $E_{S_{a,b}\leftrightarrow \{a,b\}}$ and repeat the process until no newly affected edges are found. During the process, only the edges with truss numbers in the range $[k_{min}((a,b)), \phi((a,b))]$ need to be examined. This is guaranteed by Theorem~\ref{theo:which_affected_delete}. Algorithm~\ref{algo:out_delete} gives the maintenance steps.

%To accelerate the process, we can record the local support of each edge. The local support are different for different trusses. The naive way is to record all the local support for different trusses. (may not be necessary). 

\begin{algorithm}[t]
%\SetLine
\caption{ An outwards inspection based truss maintenance algorithm for edge deletion}
\label{algo:out_delete}
{ 
\textbf{Input:} a graph $G=(V,E)$ with truss numbers associated with its edges, a to-be-deleted edge $(a,b)$; \\
\textbf{Output:} $G=(V,E$-$\{(a,b)\})$ with updated truss numbers associated with its edges;  
  \begin{algorithmic}[1]
  					\STATE delete the edge $(a,b)$;
            \STATE put the edges in $E_{S_{a,b}\leftrightarrow \{a,b\}}$ with truss numbers in $[k_{min}((a,b)), \phi((a,b))]$ (Theorem~\ref{theo:which_affected_delete}) into a queue;
                        \WHILE{the queue is not empty}
                        				\STATE take an edge $(v_1,v_2)$ from the queue;
                        				\IF{ $(v_1,v_2)$ is not marked $\wedge$ $localSupport((v_1,v_2), \phi((v_1,v_2))) < \phi((v_1,v_2))-2$ }\label{line:decrease}
                        							\STATE decrease the truss number of $(v_1, v_2)$;
                        							\STATE mark the edge $(v_1,v_2)$;
                        							\STATE put the edges in $E_{S_{v_1,v_2}\leftrightarrow \{v_1,v_2\}}$ with truss numbers in $[k_{min}((a,b)), \phi((a,b))]$ (Theorem~\ref{theo:which_affected_delete}) into the queue;
                        								
                        				\ENDIF
                        				
                        \ENDWHILE
                        \STATE start from the edge $(a,b)$, do a depth-first or breadth-first search to unmark all the marked edges; \\
                        \COMMENT {This is to guarantee the correctness of processing the next deletion or insertion.}
                        
             \RETURN $G$;
   \end{algorithmic}
}%\caption{ $decrease(G,(v_1,v_2))$}
{
\textbf{Step~\ref{line:decrease}:} $localSupport((v_1,v_2), k)$ \\ 
\textbf{Input:} an edge $(v_1,v_2)$ and a truss number $k$;\\
\textbf{Output:} the local support of the edge $(v_1,v_2)$ with respect to the truss number $k$;
  \begin{algorithmic}[1]
            \STATE $lSupport \leftarrow 0$;
            \FOR{ each node $v \in n(v_1)\cap n(v_2)$}\label{line:decreaseInter}
            		\IF{ $\phi((v,v_1)) \geq k \wedge \phi((v,v_2)) \geq k$ }
            				\STATE $lSupport \leftarrow lSupport + 1$;
            		\ENDIF
            \ENDFOR
            \RETURN $lSupport$;
   \end{algorithmic}
}

\end{algorithm}

%count how many triangle containing the edge also has the other two edges no less than $k$, this number is the local support; //compute local support;

We now explain Algorithm~\ref{algo:out_delete} in detail. Line 1 deletes the edge $(a,b)$. Line 2 puts the pruned edge set $E_{S_{a,b}\leftrightarrow \{a,b\}}$ into a queue, since this is the set of edges whose truss numbers may be directly affected by deleting $(a,b)$. While the queue is not empty (line 3), i.e. there are still some edges that may be affected directly or indirectly waiting to be inspected, we take an edge from the queue (line 4), and check whether the truss number of the edge should be decreased (line 5). The first condition ``$(v_1,v_2)$ is not marked'' means if the truss number of an edge has been decreased before (see line 7) we do not need to decrease the truss number again, because according to Lemma~\ref{lemma:how_much_effect} we only need to decrease the truss number once. The second condition of line 5 checks whether the edge $(v_1,v_2)$ has enough local support with respect to the truss number $\phi((v_1,v_2))$. If not, i.e. ``$localSupport((v_1,v_2), \phi((v_1,v_2))) < \phi((v_1,v_2))-2$'' is $true$, we decrease the truss number of $(v_1,v_2)$ (line 6), mark the edge $(v_1,v_2)$ (line 7) and put the pruned edges in $E_{S_{v_1,v_2}\leftrightarrow \{v_1,v_2\}}$ with truss numbers in $[k_{min}((a,b)), \phi((a,b))]$ 
into the queue (line 8). The truss numbers of the edges in the set $E_{S_{v_1,v_2}\leftrightarrow \{v_1,v_2\}}$ may be affected because the truss number of the edge $(v_1,v_2)$ is affected. On the contrary, if the condition ``$localSupport((v_1,v_2), \phi((v_1,v_2))) < \phi((v_1,v_2))-2$'' is $false$, no action is required. We simply go on to inspect the next edge in the queue, i.e. start a new loop in the while clause from line 4. Finally, after we empty the queue, the truss numbers of the edges in $G$ have been updated. Line 11 is a post-processing step that prepares for the next update operation.

We now explain the function $localSupport((v_1,v_2), k)$. In order to compute the local support of the edge $(v_1,v_2)$ with respect to the truss number $k$, we need to count how many triangles containing the edge $(v_1,v_2)$ also have the other two edges with truss numbers no less than $k$. This equals to count how many such node $v \in S_{v_1,v_2} = n(v_1)\cap n(v_2)$ having both $\phi((v,v_1))$ and $\phi((v,v_2))$ no less than $k$. 

%To check whether the truss number of an edge $(v_1,v_2)$ needs to be decreased, the idea is to compare the local support of the edge with the current truss number of the edge $k$. If the local support is less than the truss number $k$ minus 2 (line 8), we know the edge $(v_1,v_2)$ cannot be in a $k$-truss according to Lemma~\ref{lemma:localsupport}, so $true$ will be returned (line 9), otherwise $false$ is returned (line 11). 

In the following, we will prove Algorithm~\ref{algo:out_delete} is correct and analyze the algorithm complexity. 

\begin{lemma}
Algorithm~\ref{algo:out_delete} is correct. The complexity is $O(|E_l|)$, where $|E_l|$ is the number of edges whose local support are affected and $|E_l| << |E|$. 
\end{lemma}

\begin{proof}
The first possibly affected set of edges are put into the queue in line 2. The following possibly affected edges are put into the queue in line 8. All possibly affected edges are put into the queue by Algorithm~\ref{algo:out_delete}. When an edge $(v_1,v_2)$ is taken from the queue to be inspected, the source affected edge which leads $(v_1,v_2)$ to be put into the queue has been updated (the truss number is decreased). This guarantees the current computed local support of $(v_1,v_2)$ is an upper bound of the local support in the final updated graph. Therefore, it is safe to decrease the truss number of $(v_1,v_2)$ if the current local support is already less than $\phi((v_1,v_2))-2$. It is also guaranteed that the truss number can be decreased only once.

Complexity of Algorithm~\ref{algo:out_delete}: Assume the graph $G=(V,E)$ is stored in an adjacency list, we analyze the neighbor set intersection $n(a) \cap n(b)$ first. By using a hash table, $n(a) \cap n(b)$ can be done in $O(|n(a)|+|n(b)|)$ time. This is accomplished by firstly putting the edges in $n(a)$ into the hash table, and then hash-check the edges in $n(b)$. (An inferior method using nested-loop checks requires $O(|n(a)||n(b)|)$ time.) Back to our discussion on Algorithm~\ref{algo:out_delete}, line 1, line 2 and line 11 have complexity $O(1)$, $O(|n(a)|+|n(b)|)$ and $O(|E|)$ respectively. These are dominated by the while clause, line 3 to line 10. In the while clause, the critical step is line 5 and line 8. For line 5, the function $localSupport()$ has complexity $O(|n(v_1)|+|n(v_2)|)$, because line~\ref{line:decreaseInter} in the function $localSupport()$ has the complexity $O(|n(v_1)|+|n(v_2)|)$ and all the other steps in the function $localSupport()$ have the complexity $O(1)$. For line 8, the complexity is $O(|n(v_1)|+|n(v_2)|)$ as well. As a result, the complexity of Algorithm~\ref{algo:out_delete} is $O(n_q \cdot deg_{max})$, where $n_q$ is the number of edges ever put into the queue and $deg_{max}$ is the maximum degree of the nodes in the graph. To analyze the bounds: in a massive graph, $deg_{max}$ can be considered as a constant. We also have $n_q  = O(|E_{l}|)$, where $|E_{l}|$ is the number of edges whose local support is affected. The reason is that an edge $(a,b)$ will be put into the queue at most $2|S_{a,b}|$ times and $2|S_{a,b}| < 2|deg_{max}|$, so $n_q < 2|deg_{max}||E_{l}|$. As a result, we have $n_q = O(|E_l|)$. As the affected edges are far less than the number of edges in the graph ($|E_l| << |E|$), the maintenance algorithm is very efficient. The cost of the maintenance algorithm is linear to the number of affected edges, which is $O(E_l)$.

\end{proof} 
%$deg_{max} \leq |V|-1$. We also have $n_q < 2|V||E|$ because an edge $(a,b)$ will be put into the queue at most $2|S_{a,b}|$ times and $2|S_{a,b}| < 2|deg_{max}| < 2|V|$.   
%In the worst case, $n_q = 2(|E|-1)$ and $deg_{max} = |V|-1$, which means every edge will be put into the queue and every edge is connected with all the other edges. For example, this is the case when we delete an edge from a complete graph. 
%The complexity of Algorithm~\ref{algo:out_delete} is then $O(|E||V|^2)$. In general, $n_q$ and $deg_{max}$ are of moderate sizes, so the maintenance algorithm is efficient. 

\subsection{Inserting an edge}
\label{subsect:insert}

When an edge is inserted, the truss maintenance includes two parts. Firstly, according to Observation~\ref{lemma:obvious} and Lemma~\ref{lemma:how_much_effect}, for some existing edges, the truss numbers should be increased by 1, and all the other existing edges keep their truss numbers unchanged. 
%Therefor, we need to find the affected edges. 
Secondly, for the newly inserted edge, the truss number needs to be determined.

\subsubsection{Different from deleting an edge}

%One may think to design an approach for inserting an edge similar to the approach designed for deleting an edge, see Algorithm~\ref{algo:out_insert_false}. The idea of the hypothesized algorithm is as follows: we start from the to-be-inserted edge (line 1). Firstly, we find which neighbor edges will have their global support increased by 1 and put them into a queue (line 2). Let $(v_1,v_2)$ be one of the edges in the queue and let its truss number be $\phi((v_1,v_2))$, we then check whether the local support of $(v_1,v_2)$ with respect to the truss number $\phi((v_1,v_2))+1$ is no less than $\phi((v_1,v_2))-1$. If so, it means that $(v_1,v_2)$ has local support at least $\phi((v_1,v_2))-1$. We then increase the truss number of $(v_1,v_2)$ by 1 and go on to inspect the neighbor edges in $E_{S_{v_1,v_2}\leftrightarrow \{v_1,v_2\}}$.

One may think to design an approach for inserting an edge similar to the approach designed for deleting an edge. The idea of the hypothesized algorithm is as follows: we start from the to-be-inserted edge. Firstly, we find which neighbor edges will have their global support increased by 1 and put them into a queue. Let $(v_1,v_2)$ be one of the edges in the queue and let its truss number be $\phi((v_1,v_2))$, we then check whether the local support of $(v_1,v_2)$ with respect to the truss number $\phi((v_1,v_2))+1$ is no less than $\phi((v_1,v_2))-1$. If so, it means that $(v_1,v_2)$ has local support at least $\phi((v_1,v_2))-1$. We then increase the truss number of $(v_1,v_2)$ by 1 and go on to inspect the neighbor edges in $E_{S_{v_1,v_2}\leftrightarrow \{v_1,v_2\}}$.

The hypothesized algorithm has a similar paradigm as the deletion counterpart. However, it is not correct. The problem is that, we may miss to update some edges. In other words, after an insertion, some edges may have truss numbers increased from $k$ to $(k+1)$, but we cannot find them using the algorithm. It is safe to increase the truss number of an edge $(v_1,v_2)$ if the local support of the edge in the current $(\phi((v_1,v_2))+1)$-truss is already no less than $\phi((v_1,v_2))-1$, however, when we check the condition ``$localSupport((v_1,v_2),\phi(v_1,v_2)+1) \geq \phi((v_1,v_2))-1 $'', some edges which should finally be put into the $(\phi((v_1,v_2))+1)$-truss have not been found, so at this time, the local support may be underestimated, as a result, we may miss to update some $(\phi((v_1,v_2))+1)$-truss edges. %Let us see Fig.~\ref{} for an example.
In the following section, we introduce the correct approaches to update truss numbers in case of edge insertions.

\iffalse
\begin{algorithm}[t]
%\SetLine
\caption{ An outwards inspection based truss maintenance algorithm for edge insertion}
\label{algo:out_insert_false}
{ 
\textbf{Input:} a graph $G=(V,E)$ with truss numbers associated with its edges, a to-be-inserted edge $(a,b)$ \\
\textbf{Output:} $G=(V,E \cup \{(a,b)\})$ with updated truss numbers associated with its edges 
  \begin{algorithmic}[1]
  					\STATE insert the edge $(a,b)$;
            \STATE put the edges in $E_{S_{a,b}\leftrightarrow \{a,b\}}$ confined by Theorem~\ref{theo:which_affected_insert} into a queue;
                        \WHILE{the queue is not empty}
                        				\STATE take an edge $(v_1,v_2)$ from the queue;
                        				\IF{ $(v_1,v_2)$ is not marked $\wedge$ $localSupport((v_1,v_2),\phi(v_1,v_2)+1) \geq \phi((v_1,v_2))-1 $ }\label{line:increaseP}
                        							\STATE increase the truss number of $(v_1, v_2)$;
                        							\STATE mark the edge $(v_1,v_2)$;
                        							\STATE put the edges in $E_{S_{v_1,v_2}\leftrightarrow \{v_1,v_2\}}$ confined by Theorem~\ref{theo:which_affected_insert} into the queue;
                        								
                        				\ENDIF
                        				
                        \ENDWHILE
             \STATE compute the truss number of the inserted edge $(a,b)$;
             \STATE start from the edge $(a,b)$, do a depth-first or breadth-first search to unmark all the marked edges;
                        
             \RETURN $G$;
   \end{algorithmic}
}

\end{algorithm}
\fi

\subsubsection{Algorithm}

For truss maintenance under insertion, an intuitive, naive but correct approach is: starting from the inserted edge $e=(a,b)$, do a depth-first or breadth-first search to find the $k_{min}(e)$-truss. During the process, eliminate the edges with truss numbers more than $min(|S_{a,b}|+1,k_{max}(e))$. As a result, we can obtain a pruned $k_{min}(e)$-truss. In the end, we add the edge $e$ into the pruned $k_{min}(e)$-truss and do a truss decomposition. The truss numbers of all the edges will then be up-to-date. The correctness of this naive method is obvious, guaranteed by Theorem~\ref{theo:which_affected_insert}. However, the drawback of the naive approach is that it ignores the effect actually started from the inserted edge. In an extreme case, if we discover that the first possibly affected set of edges $E_{S_{a,b}\leftrightarrow \{a,b\}}$ are actually not affected after inserting edge $e$, we can stop the maintenance process right away, because no other edges will be affected. We only need to determine the truss number for the newly inserted edge $e$. Along this line, we design a mark-and-verify approach for truss maintenance under an edge insertion.

\begin{algorithm}[t]
%\SetLine
\caption{ A mark-and-verify truss maintenance algorithm for edge insertion}
\label{algo:out_insert}
{ 
\textbf{Input:} a graph $G=(V,E)$ with truss numbers associated with its edges, a to-be-inserted edge $(a,b)$ \\
 \textbf{Output:} $G=(V,E \cup \{(a,b)\})$ with updated truss numbers associated with its edges 
  \begin{algorithmic}[1]
  					\STATE insert edge $(a,b)$;
            \STATE put the edges in $E_{S_{a,b}\leftrightarrow \{a,b\}}$ confined by Theorem~\ref{theo:which_affected_insert} into a queue;
                        \WHILE{the queue is not empty}
                        				\STATE take an edge $(v_1,v_2)$ from the queue;
                        				\IF{ $localSupport2((v_1,v_2),\phi(v_1,v_2)) \geq \phi((v_1,v_2))-1 $ }\label{line:increaseS}
                        							\IF{ $(v_1,v_2)$ is not marked }
                        									\STATE mark the edge $(v_1,v_2)$;
                        									\STATE put the edges in $E_{S_{v_1,v_2}\leftrightarrow \{v_1,v_2\}}$ confined by Theorem~\ref{theo:which_affected_insert} into the queue;
                        							\ENDIF
                       					\ELSE
                       								\IF{ $(v_1,v_2)$ is marked }
                       										\STATE unmark the edge $(v_1,v_2)$;
                        									\STATE record the truss number of $(v_1, v_2)$ as unchanged, i.e. $unchanged((v_1, v_2))=true$;
                        									\STATE put the edges in $E_{S_{v_1,v_2}\leftrightarrow \{v_1,v_2\}}$ confined by Theorem~\ref{theo:which_affected_insert} into the queue; 
                        									%\COMMENT{the truss number of $(v_1, v_2)$ is unchanged} 
                        							\ENDIF
                        				\ENDIF
                        				
                        \ENDWHILE
                        \STATE start from the edge $(a,b)$, do a depth-first or breadth-first search: for all marked edges, unmark them and increase the truss numbers by 1;
                        \STATE compute the truss number of the inserted edge $(a,b)$;
             \RETURN $G$;
   \end{algorithmic}
}
\end{algorithm}

\begin{algorithm}[t]
\caption{ $localSupport2((v_1,v_2), k)$ }
\label{algo:localSupport2}
{
\textbf{Input:} an edge $(v_1,v_2)$ and a truss number $k$;\\
\textbf{Output:} the upperbound local support of $(v_1,v_2)$ with respect to the truss number $k$;
  \begin{algorithmic}[1]
            \STATE $lSupport \leftarrow 0$;
            \FOR{ each node $v \in n(v_1)\cap n(v_2)$}\label{line:increaseInter}
            		\IF{ $\phi((v,v_1)) \geq k \wedge unchanged((v,v_1))==false \wedge  \phi((v,v_2)) \geq k  \wedge unchanged((v,v_2))==false$ }
            				\STATE $lSupport \leftarrow lSupport + 1$;
            		\ENDIF
            \ENDFOR
            \RETURN $lSupport$;
   \end{algorithmic}
}
\end{algorithm}

The idea of mark-and-verify approach is that we determine whether the truss number of an edge should be increased in two steps. In the mark step, we try to find which edges may have their truss numbers be increased. An edge is marked if the truss number of the edge is possible to be increased by 1. As the inspection goes on, once we find an edge cannot be updated, we go on checking the neighbor marked edges to see whether there are other marked edges are affected, i.e. do not need to be updated as well. The marking steps and the verifying steps are intertwined.

The mark-and-verify algorithm is given in Algorithm~\ref{algo:out_insert}. We now explain in detail. Firstly, we insert the edge $(a,b)$ and put the edges in  $E_{S_{a,b}\leftrightarrow \{a,b\}}$ confined by Theorem~\ref{theo:which_affected_insert} into the queue (line 1 and 2). While the queue is not empty (line 3), we take an edge from the queue (line 4), and check whether the condition ``$localSupport2((v_1,v_2),\phi(v_1,v_2)) \geq \phi((v_1,v_2))-1 $ ''. If so, it means the edge $(v_1,v_2)$ may have its truss number increased by 1, because $(v_1,v_2)$ has enough local support at the moment. (Function $localSupport2()$ will be introduced in detail in the next paragraph.) As a result, if $(v_1,v_2)$ has not been identified as possibly-affected, i.e. $(v_1,v_2)$ is not marked (line 6), we mark the edge $(v_1,v_2)$ and put the edges in $E_{S_{a,b}\leftrightarrow \{a,b\}}$ confined by Theorem~\ref{theo:which_affected_insert} into the queue (line 7 and 8). If $(v_1,v_2)$ is marked, nothing needs to be done. We just leave $(v_1,v_2)$ as marked (possibly-affected). On the contrary, if ``$localSupport2((v_1,v_2),\phi(v_1,v_2)) < \phi((v_1,v_2))-1 $ '', it means the edge $(v_1,v_2)$ will keep its truss number unchanged, because it does not have enough local support with respect to the truss number $\phi((v_1,v_2)) + 1$. As a result, if $(v_1,v_2)$ is marked (identified as possibly-affected previously) in line 11, we unmark the edge $(v_1,v_2)$, record the truss number of $(v_1,v_2)$ as unchanged and put the edges in $E_{S_{a,b}\leftrightarrow \{a,b\}}$ confined by Theorem~\ref{theo:which_affected_insert} into the queue (line 12, 13 and 14). Finally, in line 18, for all the left marked edges, we increase their truss numbers by 1, because these are the affected edges. We also need to unmark them eventually to prepare for the maintenance of the next update operation. In line 19, we compute the truss number of the newly inserted edge $(a,b)$.

We now explain the function $localSupport2()$. The function takes an edge $(v_1,v_2)$ and a truss number $k$ as input. Here, $k = \phi(v_1,v_2)$. The function returns an upperbound of the local support $sup((v_1,v_2),T_{\phi(v_1,v_2)+1})$, where $sup((v_1,v_2),T_{\phi(v_1,v_2)+1})$ is the support of the edge $(v_1,v_2)$ in a $(\phi(v_1,v_2)+1)$-truss. In the beginning, all the edges with truss number $k$ have the possibility to become $(k+1)$-truss edges. As a result, $localSupport2((v_1,v_2),k)$ returns the highest bound of $sup((v_1,v_2),T_{\phi(v_1,v_2)+1})$. As the inspection goes on, some edges are found to be not affected, e.g. there may be a node $v \in n(v_1) \cap n(v_2)$ having $unchanged((v,v_1))$ changed to $true$ by Step 13 in Algorithm~\ref{algo:out_insert}. In such a case, the function $localSupport2((v_1,v_2),k)$ will return a smaller value on the next call. As more edges are found to be not affected, the value returned by the function will become smaller. Once the function $localSupport2((v_1,v_2),k)$ returns a value that is less than $k-1$, we can conclude that $(v_1,v_2)$ will not have its truss number increased by 1. Lemma~\ref{lemma:insert_mark} guarantees the conclusion.    

%Our approach is as follows: let $k$ be a truss number within the range $[k_{min}, min(k_{max},|S+1|)$, firstly start from the inserted edge $(a,b)$, check whether the edges in $E_{S\leftrightarrow \{a,b\}}$ with truss number $k$ have local support in the $k$-truss no less than $k-1$. If an edge does have no less than $k-1$ local triangles in the $k$-truss, we need to inspect the edge later. It passes the check at this time, but does not mean its truss number will be updated finally. If not, these edges cannot be in $(k+1)$-truss. This claim is supported by Lemma~\ref{lemma:insert_mark}. If the truss number of an edge is found cannot be updated, i.e. the truss number stays as $k$. We need to go back to check the undecided edges, those previously pass the check to see whether they really need to be updated. 

\begin{lemma}
Let edge $(a,b)$ have the truss number at least $k$, if the local support of $(a,b)$ in the $k$-truss containing $(a,b)$ is smaller than $(k-1)$, then edge $(a,b)$ cannot be in a $(k+1)$-truss.
\label{lemma:insert_mark} 
\end{lemma}

\begin{proof} By contradiction, assume $(a,b)$ can be in a $(k+1)$-truss, then a necessary condition is that the edge should have local support at least $k-1$ in the $(k+1)$-truss. Since a $(k+1)$-truss is also a $k$-truss, it implies that $(a,b)$ must have at least $k-1$ support in the $k$-truss containing $(a,b)$. A contradiction appears. 
%In the extreme case, all $k$-truss edges are $(k+1)$-truss edges. In such case, if the edge $(a,b)$ does not have enough local support, then the edge cannot be in the $(k+1)$-truss. A necessary condition for an edge to be in a $(k+1)$-truss is that the edge should have local support at least $k-1$.
\end{proof}

%Another important aspect of the mark-and-verify method, compared with the naive decomposition based method, is that we can inspect different truss numbers simultaneously. Lemma~\ref{lemma:simultaneous} is supporting this.

%\begin{lemma}
%Following the mark-and-verify approach, truss maintenance can be done simultaneously for different truss numbers. 
%\label{lemma:simultaneous}
%\end{lemma}

%\begin{proof} According to Lemma~\ref{lemma:how_much_effect}, edges with truss number no more than $k-1$ cannot be in $(k+1)$-trusses after one edge insertion. As a result, it is safe to inspect only $k$-truss edges to find the new $(k+1)$-truss edges. The condition holds for every $k$ between $[k_{min}, min(k_{max},|S+1|)$.
%\end{proof}

In the following, we prove Algorithm~\ref{algo:out_insert} is correct and analyze the algorithm complexity.

\begin{lemma}
Algorithm~\ref{algo:out_insert} is correct. The complexity is $O(|E_l|)$, where $|E_l|$ is the number of edges whose local support are affected and $|E_l| << |E|$.
\end{lemma}

\begin{proof}
The first possibly-affected set of edges are put into the queue in line 2. For the other edges, the local support condition ``$localSupport2((v_1,v_2),\phi(v_1,v_2)) \geq \phi((v_1,v_2))-1 $'' and line 6-8 guarantee that all possibly-affected edges will be put into the queue. ``$localSupport2((v_1,v_2),\phi(v_1,v_2)) < \phi((v_1,v_2))-1 $'' and line 11-14 guarantee that once an edge is found to be not affected. The effect will be spread to the full set of other edges.  
%Obviously, the naive truss decomposition method is correct. Algorithm~\ref{algo:out_insert} follows the same idea and can be considered as a guided version of truss decomposition. The inspection started from the inserted edge rather than an arbitrary one. Lemma~\ref{lemma:simultaneous} guarantees the update can be done simultaneously for each $k$.

Complexity of Algorithm~\ref{algo:out_insert}: similar to Algorithm~\ref{algo:out_delete}, the dominating step of Algorithm~\ref{algo:out_insert} is the while clause, from line 3 to line 17. In the while clause, the critical steps are line 5, line 8 and line 14. All of them have complexity $O(|n(v_1)|+|n(v_2)|)$. Again, the complexity of Algorithm~\ref{algo:out_insert} is $O(n_q \cdot deg_{max})$. Different from the deletion case, $n_q$ is bounded by $4|deg_{max}||E_l|$ because the edges in $E_{S_{v_1,v_2}\leftrightarrow \{v_1,v_2\}}$ will be put into the queue when we both mark and unmark the edge $(v_1,v_2)$, while for the deletion algorithm, we only put the edges in $E_{S_{v_1,v_2}\leftrightarrow \{v_1,v_2\}}$ into the queue when we mark the edge. Nevertheless, the worst case complexity of Algorithm~\ref{algo:out_insert} is also $O(|E_l|)$, where we usually have $|E_l| << |E|$.
\end{proof}

\section{Query index maintenance}
\label{sect:index}

The ultimate purpose of maintaining truss numbers for each edge is to serve queries. In the real world, a graph is often involving and thus the trusses are changing from time to time. At a particular time, if a user wants to know the current $k$-trusses, it is desired we can answer this query as soon as possible. In this section, we will introduce how maintaining truss numbers for each edge can help answer the queries quickly, and what other maintenance work we need to do, i.e. query index maintenance.

\subsection{Querying the $k$-trusses}
To find the $k$-trusses from a graph, if we do not have the truss numbers recorded on the edges, we may need to recompute the $k$-trusses from scratch if the graph has been updated since the last query, because the $k$-trusses may change after the updates. Apparently, recomputing the $k$-trusses is very time-consuming, even though some pruning methods can be designed to speed up the process, e.g. pruning the edges that have global support less than $k-2$. On the other hand, if we have the truss numbers recorded and maintained, things will be easier. A simple method is to traverse the graph and group connected edges with truss number $k$ together. Each connected component is then a $k$-truss.

\subsection{Indexing the $k$-trusses}

To avoid traversing the whole graph to answer a $k$-truss query, we can take one edge from a $k$-truss and represent the $k$-truss with the edge which we call a representative. In this way, the $k$-truss containing the representative can be found by performing a depth-first or breadth-first traversal from the representative on the graph. To be specific, let the representative be $e$, in the traversal, when we meet other edges with truss number $\phi(e)$, we include the new edges into the truss containing the representative, and finally we can find a truss with truss number $\phi(e)$. As a result, if we want to know all the $k$-trusses, we just need to find out $k$ edges each of which is in a standalone truss, and then find the trusses by traversing the graph from those edges. That means if there are $t$ $k$-trusses, we can index $t$ edges to represent the $t$ $k$-trusses. 

How to choose a representative? In general, a representative can be chosen randomly. It takes similar cost to find a $k$-truss starting from different representatives within the $k$-truss. 

\subsection{Maintaining the index}

In this section, we introduce how to maintain the index, i.e. maintain the representatives.

%\subsubsection{Lazy maintenance}
%The idea of lazy maintenance is straightforward. That is we update the representatives after all the truss numbers of the affected edges are updated. 
\subsubsection{Deletion}

Firstly, we discuss a deletion. If a deletion occurs, a number of edges will be affected from the deleted edge $e$. For a particular truss number $k \in [k_{min}(e),\phi(e)]$, there are two types of change related to $k$-trusses: (1) some $(k+1)$-truss edges become $k$-truss edges; (2) some $k$-truss edges become $(k-1)$-truss edges. As to the representatives, there are three cases: (1) some $k$-truss representatives may stay as representatives because their truss numbers stay the same; (2) some $k$-truss representatives may stop being representatives because their truss numbers change to $k-1$; (3) in addition, we may need to add some new representatives into the index to represent some newly generated $k-1$-trusses.  

Consider the original $k$-truss $T_{k}$ containing the deleted edge $e$, after deleting $e$, we may further delete some other edges (starting from $e$) whose truss numbers downgraded to $k-1$ in $T_{k}$. As a result, the updated $T_{k}$ may shrink into a smaller $k$-truss, or may be decomposed into several $k$-trusses, or may even disappear ($T_{k}$ becomes a part of a $(k-1)-truss$). To maintain the representatives for the $k$-trusses, the method is to search the updated $T_{k}$ and choose a representative for each connected component. If a representative is not in the current $k$-truss index, we add the new representative into the index. If a previous representative has the truss number downgraded from $k$ to $k-1$, then remove the representative from the $k$-truss index. In this way, truss index is maintained for deletion.

%We have the following property: if an edge is affected, i.e. the truss of an edge is reduced from $k$ to $k-1$, the edge must be connected with another affected edge also with truss number reduced from $k$ to $k-1$ or connected with the deleted edge. This property can be proved.

For each edge, we can record its representative. For the edges in the non-affected area, if their representatives are deleted, we then choose themselves as representatives, and propagate the information onwards until all the bridges between the non-affected area and the affected area have been examined.  For the affected area, we explore from the deleted edge to find a connected affected area with truss number $k-1$ (the affected area had truss number $k$ before the deletion). If this area is connected with an edge represented by an edge $e_{rep}$ with unchanged truss number $k-1$, then all the edges in the connected affected area will be represented by the representative $e_{rep}$, otherwise, we randomly choose an edge from the connected affected area as a new representative to represent the area. 

\subsubsection{Insertion}

Secondly, we discuss an insertion. If a insertion occurs, similarly a number of edges will be affected from the inserted edge $e=(a,b)$. For a particular truss number $k \in [k_{min}(e),min(|S_{a,b}|+1, k_{max}(e))]$, there are also two types of change related to the $k$-truss: (1) some $k$-truss edges become $(k+1)$-truss edges; (2) some $(k-1)$-truss edges become $k$-truss edges. For the representatives, (1) some $k$-truss representatives may stay as representatives; (2) some existing $k$-truss representatives may be removed, because, after inserting the edge $e$, some existing $k$-trusses can be merged into a larger $k$-truss, so we only need one edge to represent the larger $k$-truss; finally, we may need to add some new $(k+1)$-truss representatives if there are newly generated $(k+1)$-truss.   

Again, let $T_{k}$ be a $k$-truss before inserting $e$, where $k$ is in the range $[k_{min}(e),min(|S_{a,b}|+1, k_{max}(e))]$, (1) if $k > \phi(e)$, the $k$-truss representatives remain the same, because $e$ is not in any $k$-truss, and thus does not affect the $k$-truss indexes; (2) if $k \leq \phi(e)$, we start from the edge $e$ and traverse the subgraph $T_{k}\cup\{e\}$ by inspecting the connected edges with truss number $k$. During the process, if we meet more than one existing representatives, we keep one of them in the index and delete other representatives from the index. If we do not meet any existing representative, we add one of the edges with truss number $k$ into the index to represent the current $k$-truss. In this way, truss index is maintained for insertion.
%(1) If the affected edges does not include any representative edges, start from the affected edges from each affected truss number $k' \in [k_{min}, min\{k_{max}, |S|+1\}]$, perform a traversal, if the traversal meets a representative edge, then disregard these affected edges, because they can be found by traversing the graph through the representative. Otherwise, begin with higher trusses, traverse the graph from the promoted representative (from $k$ to $k+1$), if an existing $k+1$-truss representative is met, it means the updated edges can be found through an existing $k+1$-truss, we can disregard the updated edges; otherwise we choose an edge as a representative from the updated edges to represent the new $(k+1)$-truss. At the same time, remove the promoted $(k+1)$-truss edges from the existing $k$-truss and choose a new representative for the existing $k$-truss. 

Similarly, for a non-affected area, the truss numbers of the edges are not changed. The edges will be represented by the same edge representatives even though the edge representatives may upgrade from $k$-trusses to $k+1$-trusses. This is because  a $(k+1)$-truss representative can also be a $k$-truss representative, since the representative can be guaranteed to be in a $k$-truss. There may be repeated representatives representing the same connected truss subgraph. The repetition will be resolved when querying the $k$-trusses when repeated representatives will be removed. For the affected area, for $k \in [k_{min}(e), min(|S_{a,b}|+1, k_{max}(e))]$, we do the following: if an affected edge $e'$ has the truss number increased to $k+1$, but the truss number of its representative is still $k$, we change the truss representative of $e'$ into $e'$ itself, and then propagate the effect to the neighbor edges.

\iffalse

Summarize: representative maintenance. For deletion, (1) stay as rep; (2) not rep any more (because of truss reduce); (3) add a new rep (because of truss split). similarly, for insertion, (1) stay as rep; (2) not a rep any more (because of truss merge); (3) add a new rep (because of new truss).

Summarize: truss change. Current truss containing $(a,b)$ will be: for deletion, (1) smaller truss; (2) several trusses; (3) disappear. for insertion, (1) remain the same; (2) two trusses merges; (3) a new $(k+1)$-truss produced.

How to combine the maintenance process seaminglessly with updates.

\fi

\section{Experiments}
\label{sect:experiments}

In this section, we report the performance of the truss maintenance algorithms. All experiments are done on a desktop with Intel(R) Core(TM) i5-2400 CPU at 3.1GHz and 8GB RAM. The operating system is Windows 7 Enterprise, and the code is written in Java.

\subsection{Datasets and Approaches}
We use three datasets to test the algorithms, Epinions social network (soc-Epinions1) \cite{DBLP:conf/semweb/RichardsonAD03}, Enron email network (email-Enron)~\cite{DBLP:journals/corr/abs-0810-1355}, and Slashdot social network (Slashdot0811)~\cite{DBLP:journals/corr/abs-0810-1355}. Epinions social network is a who-trust-whom online social network of a general consumer review site (www.Epinions.com). Members of the site decide whether to trust each other. Enron email communication network covers all the email communication within a dataset of around half million emails, which are collected and published by Federal Energy Regulatory Commission. Slashdot social network is a network that contains friend links between the users of Slashdot. All the above datasets are in Stanford Large Network Dataset Collection\footnote{http://snap.stanford.edu/data/}. We give the details of each dataset in terms of number of vertices, edges, and average degrees in Table~\ref{tab:dataset}. We randomly generated enough number of updates (insertions and deletions) for each dataset and stored them for reuse. We guarantee that when we test different approaches on a particular dataset, we use the same set of updates. The approaches we have compared are listed in Table~\ref{tab:vertexApproach}

\begin{table}[t]
\renewcommand{\arraystretch}{1.3}
\centering
\caption{Datasets}
\label{tab:dataset}
  \begin{tabular}{c|c|c|c}
  \hline
  \hline
  & vertices & edges & maximum truss \\
  \hline
  Epinions network & 75879 & 508837 & 33 \\ 
  
  \hline
  Enron email network & 36692 & 183831 & 22 \\
  \hline
  Slashdot social network & 77360 & 905468 & 34 \\
  \hline
  \hline
  \end{tabular}
\end{table}

\begin{table}[t]
\renewcommand{\arraystretch}{1.3}
\centering
\caption{Summary of the approaches}
\label{tab:vertexApproach}
  \begin{tabular}{l|l}
  \hline
  \hline
  batchUpdate 			& process the updates in a batch mode and \\
  									& do a truss decomposition to find the new \\
  									& $k$-truss \\
  \hline
  progressiveUpdate & maintain the truss number for each edge, \\ 
  									& process the updates on-the-fly and search\\
  									& the graph to find the new $k$-truss \\
  									& using the recorded truss numbers \\  
  \hline
  indexedUpdate 		& maintain the truss number for each edge, \\ 
         						& maintain  truss representatives as an index, \\
         						& process the updates on-the-fly and find \\
         						& the new $k$-truss by searching the graph \\
         						& from the representatives \\
  \hline
  \hline
  \end{tabular}
\end{table}

\subsection{Effect of batchUpdate vs progressiveUpdate}
In this section, we compare the time cost of batchUpdate and progressiveUpdate. The batchUpdate approach stands for processing the updates in a batch and then doing truss decomposition on the graph to find the $k$-truss from the graph. The decomposition starts from smaller truss numbers and stops at the number $k$. The progressiveUpdate approach means that the algorithms introduced in Section~\ref{sect:algorithms} are used to maintain truss numbers while performing updates and the $k$-trusses are evaluated by searching the graph using the truss numbers. Fig~\ref{fig:epinion} shows that progressiveUpdate is much better than batchUpdate when the number of updates is moderate for all the datasets. As the number of updates increase, the gap between batchUpdate and progressiveUpdate decreases. The cost of progressiveUpdate will catch up with batchUpdate when there are 10,000 updates. Fig.~\ref{fig:email} and~\ref{fig:slashdot} have similar trends, but the catch-up points are different.

\begin{figure*}[t]
  \centering
  \subfigure[$k=30$]{\label{fig:epinion30}
    \includegraphics[height=48mm,width=60mm]{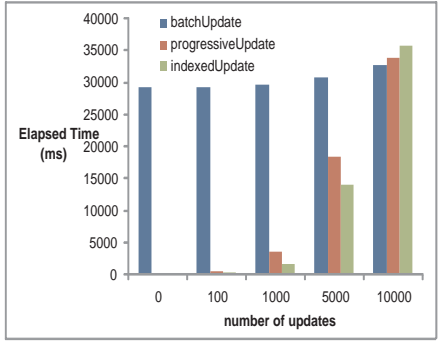}}
 \hskip 0.5in
  \subfigure[$k=25$]{\label{fig:epinion25}
    \includegraphics[height=48mm,width=60mm]{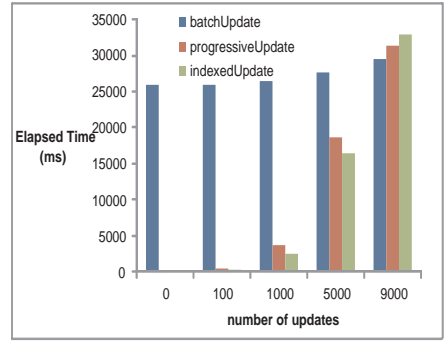}} 
 \hskip 0.5in
  \subfigure[$k=20$]{\label{fig:epinion20}
    \includegraphics[height=48mm,width=60mm]{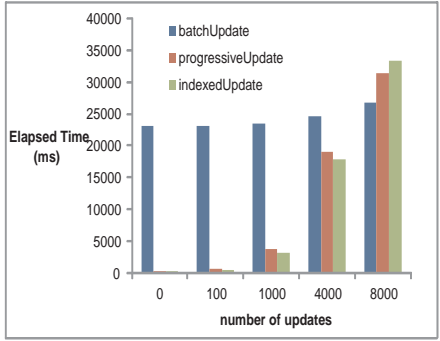}}
  \hskip 0.5in
  \subfigure[$k=15$]{\label{fig:epinion15}
    \includegraphics[height=48mm,width=60mm]{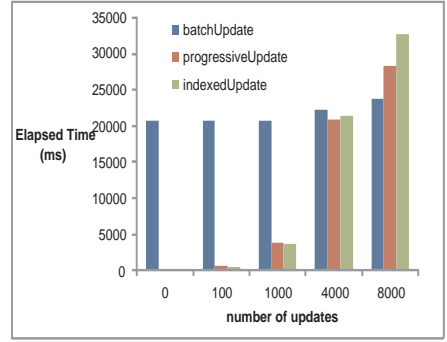}}    
\caption{Epinions Social Network}
\label{fig:epinion}
\end{figure*}

\subsection{Effect of index vs non-index}
In this section, we compare the performance of progressiveUpdate and indexedUpdate. Compared with the progressiveUpdate approach, indexedUpdate additionally builds and maintains the truss indexes, i.e. truss representatives. Queries will be answered by traversing the graph from the representatives and the representatives are maintained when processing edge updates. For most datasets, it can be seen that indexedUpdate is better than progressiveUpdate when the number of updates is moderate. As the number of updates becomes larger, indexedUpdate becomes worse than progressiveUpdate, such as the result shown in Fig.~\ref{fig:epinion15},~\ref{fig:email10} and~\ref{fig:slash15}. The reason is probably that indexedUpdates spends more time on maintaining the indexes.

\subsection{Effect of different truss number $k$}
In this section, we report the effect of different truss number $k$. As we can see, for all datasets, for larger numbers of $k$'s, progressiveUpdate and indexedUpdate are much better batchUpdate, i.e. they can accommodate more updates before they become worse than batchUpdate. For example, for Enron email data, 8,000 updates make progressiveUpdate worse than batchUpdate for $k$=18 (see Fig.~\ref{fig:email18}), while only 6,000 updates make progressiveUpdate worse (see Fig.~\ref{fig:email10}). The reason is that, for a larger $k$, only a small portion of the graph needs to be accessed when processing edge updates. A larger $k$ implies a smaller $k$-truss. For smaller $k$'s, the advantage of progressiveUpdate and indexedUpdate compared to batchUpdate become less obvious.

\subsection{Effect of different datasets}

In this section, we report the effect of different datasets. In Fig.~\ref{fig:slashdot}, for Slashdot data (average clustering coefficient: 0.0555), progressiveUpdate and indexedUpdate are better than batchUpdate up to 16,000 updates for $k$=30, while for Epinions data (average clustering coefficient: 0.1378), progressiveUpdate and indexedUpdate are only better than batchUpdate under around 10,000 updates (Fig.~\ref{fig:epinion30}). The reason is that Slashdot data is relatively sparser and is not a centralized dataset. As a result, the updates are usually done less costly. Epinions data is more centralized than Slashdot data, and thus each update needs to access more nodes. Among the three datasets, Enron email network (average clustering coefficient: 0.4970) is the smallest dataset, on which all the approaches run faster, however, progressiveUpdate and indexedUpdate allow the smallest number of updates before perform worse than the naive batchUpdate approach. Our maintenance approaches work better for decentralized data.

\begin{figure*}[t]
  \centering
  \subfigure[$k=22$]{\label{fig:email22}
    \includegraphics[height=48mm,width=60mm]{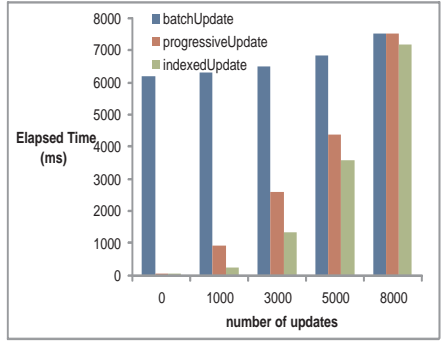}}
 \hskip 0.5in
  \subfigure[$k=18$]{\label{fig:email18}
    \includegraphics[height=48mm,width=60mm]{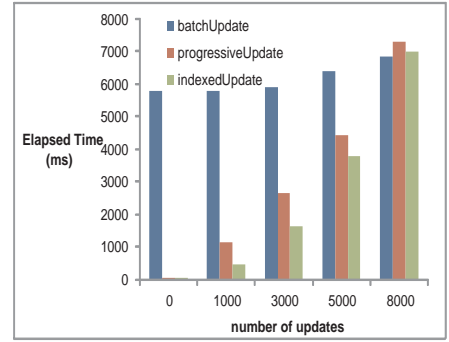}} 
 \hskip 0.5in
  \subfigure[$k=14$]{\label{fig:email14}
    \includegraphics[height=48mm,width=60mm]{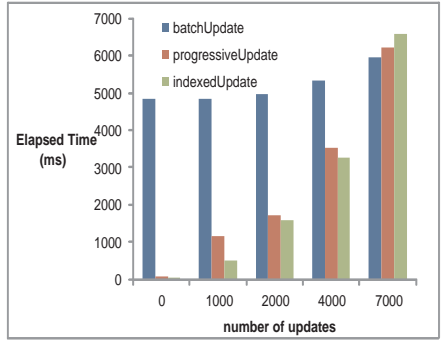}}
 \hskip 0.5in
  \subfigure[$k=10$]{\label{fig:email10}
    \includegraphics[height=48mm,width=60mm]{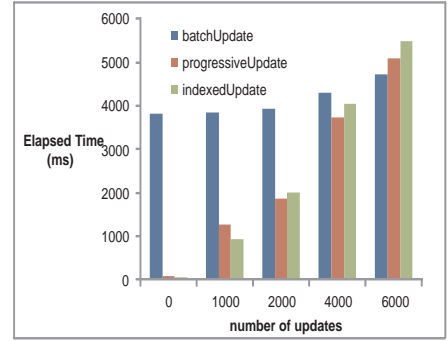}}    
\caption{Enron Email Network}
\label{fig:email}
\end{figure*}

\begin{figure*}[t]
  \centering
  \subfigure[$k=30$]{\label{fig:slash30}
    \includegraphics[height=48mm,width=60mm]{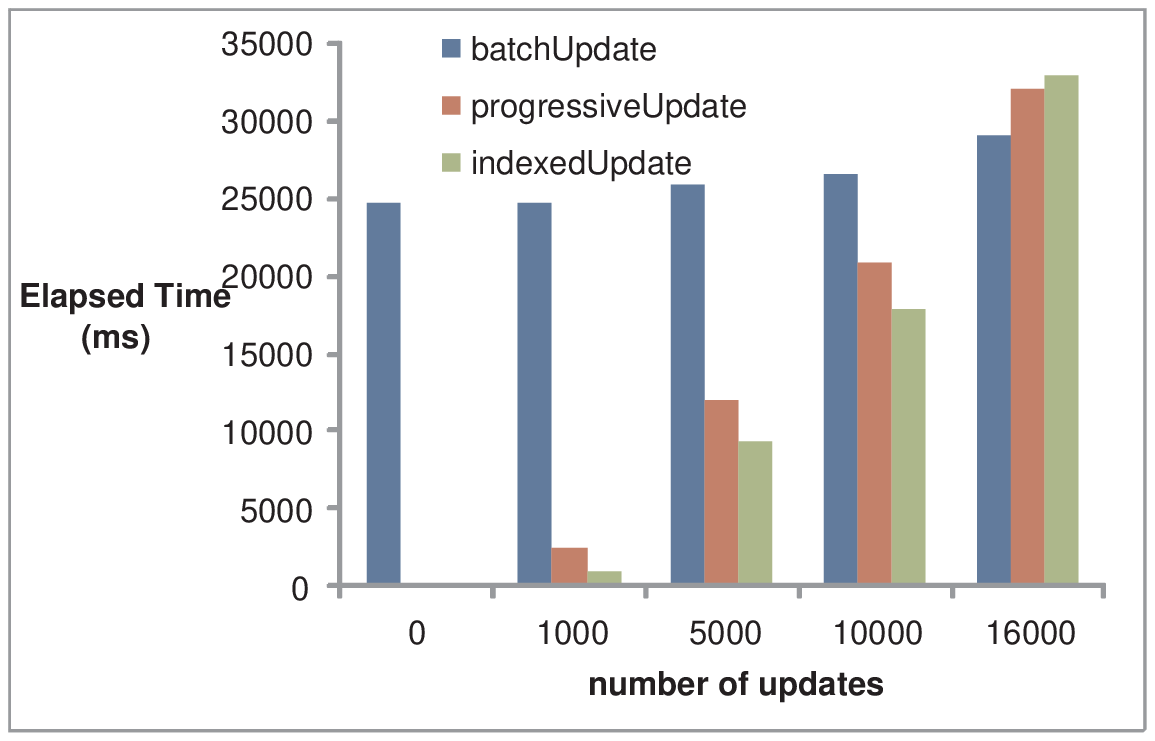}}
 \hskip 0.5in
  \subfigure[$k=25$]{\label{fig:slash25}
    \includegraphics[height=48mm,width=60mm]{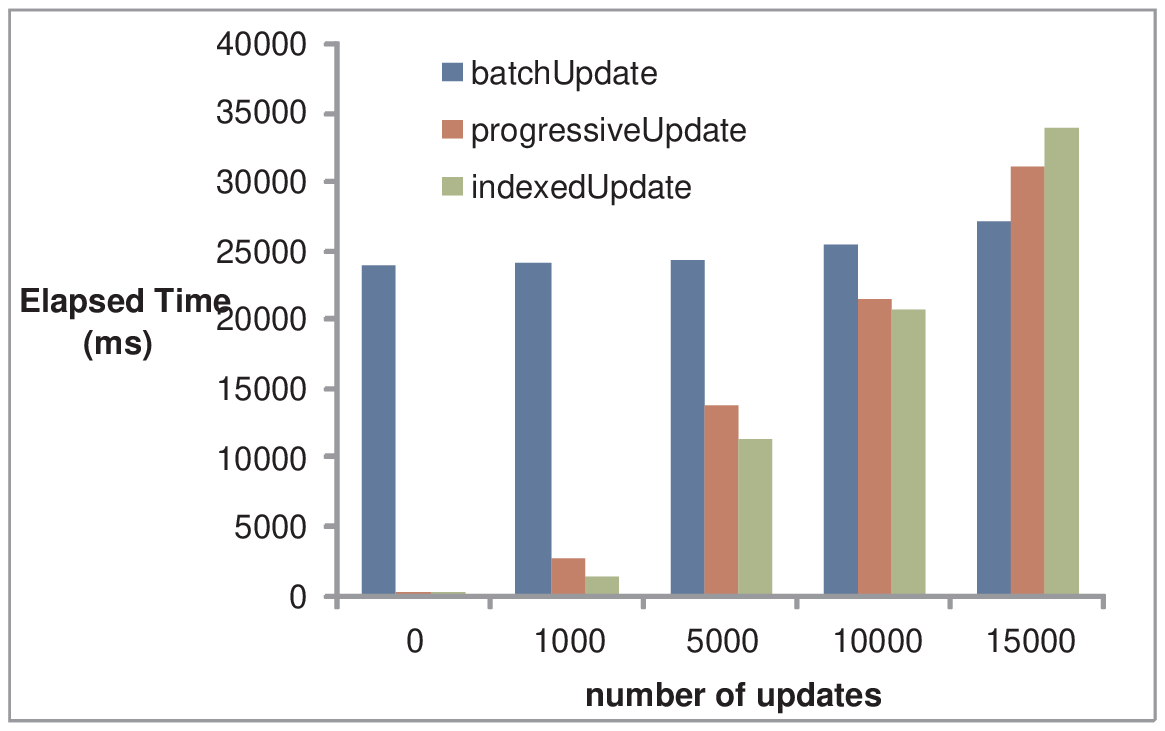}} 
 \hskip 0.5in
  \subfigure[$k=20$]{\label{fig:slash20}
    \includegraphics[height=48mm,width=60mm]{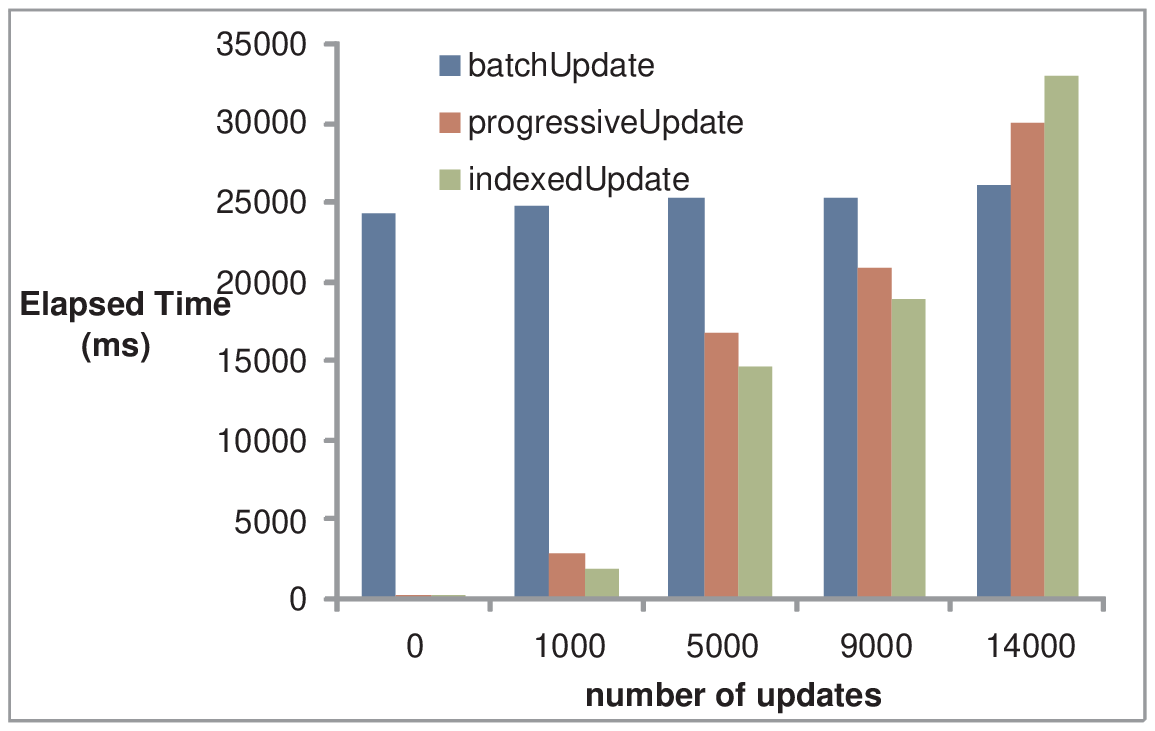}}
 \hskip 0.5in
  \subfigure[$k=15$]{\label{fig:slash15}
    \includegraphics[height=48mm,width=60mm]{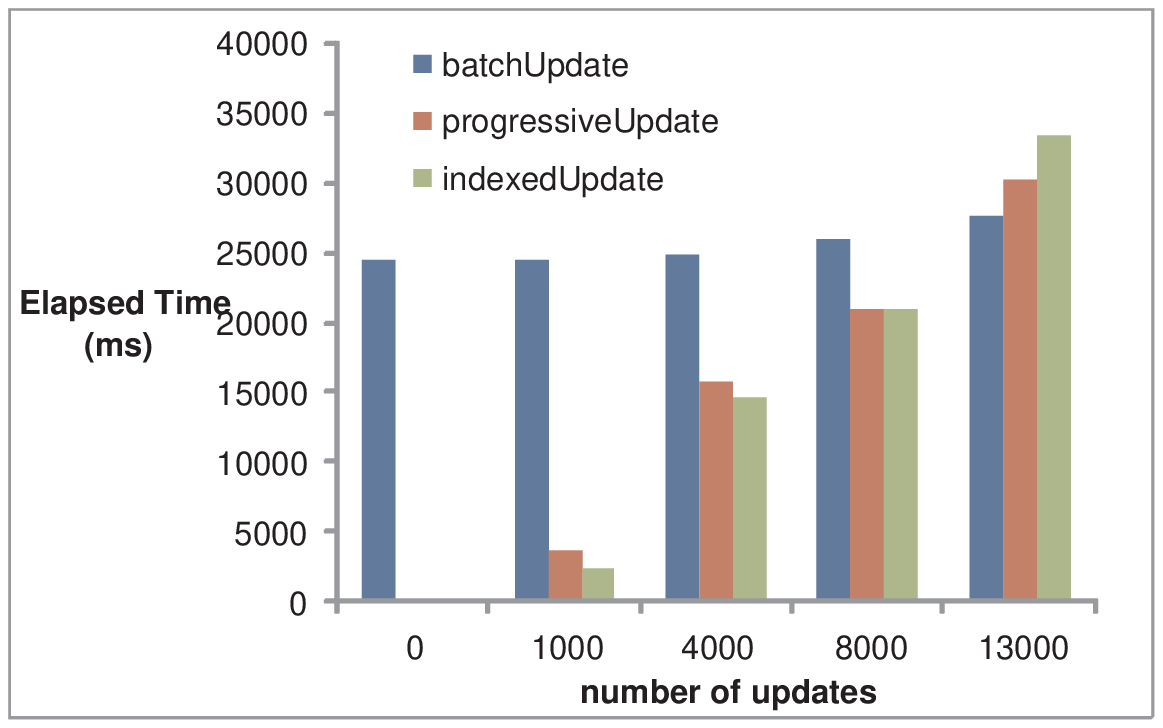}}    
\caption{Slashdot Social Network}
\label{fig:slashdot}
\end{figure*}

\section{Related Work}
\label{sect:related}
 
The concept of ``truss'' was firstly introduced by Cohen~\cite{DBLP:journals/web/truss} in 2008, when social networks developed fast and corresponding research prevailed. Compared with other cohesive subgraph models, such as clique~\cite{RePEc:spr:psycho:v:14:y:1949:i:2:p:95-116}, quasi-clique~\cite{DBLP:conf/latin/AbelloRS02, DBLP:journals/tods/ZengWZK07}, $n$-clique~\cite{n_clique}, $n$-clan~\cite{n_clan}, $n$-club~\cite{n_clan}, $k$-plex~\cite{Seidman1978}, $k$-core~\cite{Seidman1983269}, $k$-truss has its own advantage (mentioned in the introduction already). The computation of $k$-truss, also called truss decomposition, has been studied in~\cite{DBLP:journals/web/truss,Cohen:2009:GTM:1591877.1592018,Wang:2012:TDM:2311906.2311909}. Cohen proposed the first truss-decomposition algorithm~\cite{DBLP:journals/web/truss}, which is later outperformed by an improved in-memory algorithm proposed by Wang and Cheng~\cite{Wang:2012:TDM:2311906.2311909}. Wang and Cheng proposed an out-of-memory algorithm for truss decomposition and a top-$t$ $k$-truss evaluation algorithm~\cite{Wang:2012:TDM:2311906.2311909}. Cohen~\cite{Cohen:2009:GTM:1591877.1592018} also proposed a map-reduce based algorithm for truss decomposition. Recently, Zhao and Tung~\cite{DBLP:journals/pvldb/ZhaoT12} studied the truss decomposition problem and consider that the networked data is stored in a graph database. They also studied how to visualize the graph. Different from all the above works, we do not consider finding the trusses from scratch. We aim to maintain the trusses in face of frequent updates. 

Another closely related work to ours is $k$-core maintenance~\cite{DBLP:journals/corr/abs-1207-4567}. Although the ``maintenance'' theme is shared in these two works, the technical solutions are different, because $k$-core and $k$-truss impose requirements on different aspects. To be specific, $k$-core imposes requirements on the nodes, asking each node to connect to at least $k$ other nodes; while $k$-truss imposes requirements on the edges, asking each edge to be in at least $k-2$ triangles. Moreover, in our work, we also discussed how to build $k$-truss indexes and how to maintain $k$-truss indexes. These corresponding questions were not discussed for the $k$-core counterpart in the work~\cite{DBLP:journals/corr/abs-1207-4567}.    
 
All the models above can be considered as explicit models for discovering cohesive structures from networked data. The common feature is that a structure with certain property is predefined, and then the rest work is to design efficient algorithms to discover all the subgraphs with the structure requirement. Another stream is called implicit models, some proposed objective functions first, such as modulariy~\cite{newman-2004-69}, normalized cut~\cite{DBLP:journals/pami/ShiM00}, and then partitioned the graph into a number of parts where the objective functions can be maximized or minimized; Some works defined neighbourhood distance, such as propinquity~\cite{DBLP:conf/kdd/ZhangWWZ09}, structure closeness~\cite{DBLP:conf/kdd/XuYFS07}, and then grouped nearby nodes within a distance threshold around a given node to form a group; Some borrowed the idea of Markov Clustering~\cite{DBLP:conf/kdd/SatuluriP09, DBLP:journals/jgaa/PonsL06} to repeat \emph{random walk} for a few rounds until self-organized clusters turn up. The maintenance problem for implicit models are not studied as well.

\section{Conclusions and Future Work}
\label{sect:conclusion}

In this paper, we have studied how to maintain trusses in an evolving network by considering edge deletions and insertions. For one edge deletion or insertion, we have investigated the properties of truss change, proposed algorithms to perform truss maintenance and introduced techniques to maintenance truss index updates. In the experiments, we have shown that maintaining truss on-the-fly is more efficient compared to performing batch edge updates and then doing truss decomposition if the update number is not large. We have also shown that having truss indexes and maintaining truss indexes are superior than evaluating truss queries without index support for moderate update frequency. As far as we can see, one direction for future work is to process data that cannot fit into memory, e.g., a $k$-truss may be very large when $k$ is small.   
\bibliographystyle{unsrt}
\bibliography{sigproc}  % sigproc.bib is the name of the Bibliography in this case

\end{document}